\documentclass[pdflatex,sn-mathphys-num]{sn-jnl}

\usepackage{graphicx}%
\usepackage{multirow}%
\usepackage{amsmath,amssymb,amsfonts}%
\usepackage{amsthm}%
\usepackage{mathrsfs}%
\usepackage[title]{appendix}%
\usepackage{xcolor}%
\usepackage{textcomp}%
\usepackage{manyfoot}%
\usepackage{booktabs}%
\usepackage{algorithm}%
\usepackage{algorithmicx}%
\usepackage{algpseudocode}%
\usepackage{listings}%

\theoremstyle{thmstyleone}%
\newtheorem{theorem}{Theorem}
\newtheorem{proposition}[theorem]{Proposition}%

\theoremstyle{thmstyletwo}%
\newtheorem{remark}{Remark}%

\theoremstyle{thmstylethree}%
\newtheorem{definition}{Definition}
\newtheorem{assumption}{Assumption}
\newtheorem{lemma}[theorem]{Lemma}

\newcommand{\Exit}{\partial_{\mathrm{out}}}

\begin{document}

\title{Entropic Risk-Sensitive Evolutionary Learning and Equilibrium Selection in Coordination Games}


\author[1]{\fnm{Solaleh} \sur{Mohammadi}}\email{solalehm@umd.edu}

\author[2]{\fnm{Xiang} \sur{Gao}}\email{xgao242424@gmail.com}

\author*[1]{\fnm{Kaiqing} \sur{Zhang}}\email{kaiqing@umd.edu}

\affil[1]{\orgname{University of Maryland, College Park}}

\affil[2]{\orgname{University of Illinois Urbana-Champaign}} 



\abstract{
We study risk-sensitive evolutionary learning dynamics and their long-run equilibrium selection behaviors in coordination games. 
Agents' risk attitudes enter through the classical entropic risk measure, which evaluates opponent-induced payoff uncertainty and feeds into noisy best responses under two standard revision protocols: \emph{best response with mutations} and \emph{logit choice}. We first analyze $2\times 2$ coordination games in both single-population symmetric and two-population asymmetric settings. In the single-population setting, unlike the risk-neutral case where the dynamics are known to favor the \emph{risk-dominant equilibrium}, we show that risk sensitivity can change the stochastically stable outcome: a greater risk-seeking attitude favors the \emph{payoff-dominant equilibrium}, while a greater risk-averse attitude favors the \emph{maximin equilibrium}. 
Thus, the population's risk attitude may act as a control knob for long-run equilibrium selection. In both population settings, we also identify a robust regime: any \emph{super-dominant}  equilibrium is stochastically stable for \emph{all} risk attitudes, under both protocols, and across populations.  
We further extend the single-population analysis to symmetric $k$-action games, which include  symmetric $k$-action coordination games as a special case, under risk-sensitive best response with mutations. In this setting, we show that, for sufficiently large populations, sufficiently risk-seeking agents uniquely select the \emph{strongly payoff-dominant equilibrium}  when it exists, whereas sufficiently risk-averse agents uniquely select the \emph{strongly maximin equilibrium} when it exists. These results show that entropic risk sensitivity may serve as a systematic mechanism for steering equilibrium selection in evolutionary games,  beyond the classical risk-neutral benchmark.
}

\keywords{Risk Sensitivity, Evolutionary Dynamics, Coordination Games, Equilibrium Selection}


 
\maketitle

\section{Introduction}\label{sec1}


In game theory, strategic interactions and learning processes among agents often come with substantial uncertainty, driven by opponents' behavior, environmental randomness, and/or noise in payoffs and decision-making. For this reason, earlier literature argued that it is unrealistic to assume agents are always \emph{risk-neutral} against such uncertainties, and behave as \emph{expectation maximizers}  \cite{goeree2003risk, fiat2010players}. Instead, agents may be \emph{risk-averse} or \emph{risk-seeking}, and such attitudes are commonly modeled using risk measures that go beyond expected payoffs.
More recently, risk sensitivity has also been incorporated in game theory for more tractable equilibrium computation \cite{mazumdar2024tractable} and for promoting cooperation in multi-agent systems \cite{noorani2022risk,zhang2025optimism}. In particular, these works studied games whose payoff or value functions are \emph{transformed}  by risk sensitivity, and then analyzed the existence \cite{fiat2010players} and computation \cite{mazumdar2024tractable,noorani2022risk,zhang2025optimism} of equilibria in the resulting \emph{transformed} games.

What remains less understood is whether \emph{risk attitudes} can be incorporated directly into agents' \emph{learning dynamics}, and how doing so changes \emph{equilibrium selection} in the \emph{original} game, which asks 
the question: 
{which outcome will agents coordinate on through learning, and why  \cite{harsanyi1988general,sandholm2010population}?} 
Towards addressing this question, we study the stochastic stability of risk-sensitive evolutionary dynamics, in caonical $2\times2$ coordination games \cite{harsanyi1988general,kandori1993learning,young1993evolution}. We consider both the single-population symmetric setting and the two-population asymmetric setting, allowing different risk attitudes across populations. At each step, one random agent revises its action according to a risk-sensitive noisy best-response rule. Risk sensitivity enters through the classical entropic risk measure (ERM) \cite{follmer2002convex}, which replaces the risk-neutral expected payoffs used in learning. We focus on two standard revision protocols from evolutionary game theory: \emph{best response with mutations} (BRM) \cite{kandori1993learning} and \emph{logit choice}  \cite{blume1993statistical}.

Prior work in evolutionary game theory showed that, with risk-neutral agents, the risk-dominant equilibrium is often uniquely selected in the long run \cite{kandori1993learning,young1993evolution,staudigl2012stochastic}. Few results have studied the impact of risk sensitivity on evolutionary learning, except, e.g., \cite{sawa2018reference,nax2019risk}. 
They study how risk dominance and stochastic stability are \emph{preserved} under reference-dependent preferences \cite{sawa2018reference} and different risk attitudes \cite{nax2019risk}, particularly when risk dominance {coincides} with maximin or payoff dominance. By comparison, our results systematically characterize \emph{how} varying entropic risk sensitivity in the evolutionary dynamics \emph{changes} long-run equilibrium selection.

More precisely, in the single-population symmetric setting, and under both protocols, when the payoff-dominant and maximin equilibria differ, increasing the risk-sensitivity parameter shifts stochastic stability from the maximin equilibrium to the payoff-dominant equilibrium. This is also consistent with the recent empirical observations that more risk-seeking learning agents are more likely to \emph{coordinate} at the higher-payoff equilibrium \cite{noorani2022risk,zhang2025optimism},
and our results can thus be viewed as a formalization of these insights through the lens of coordination games and evolutionary dynamics. 
By contrast, if there exists a {super-dominant} equilibrium \cite{sawa2018reference}, which is both weakly maximin and weakly payoff-dominant, then it is stochastically stable for all risk attitudes under both protocols and in both population settings. 
Compared with  \cite{sawa2018reference,nax2019risk}, our insight on ERM aligns with theirs in the single-population BRM setting, regarding the robustness of super-dominance, although we start from a different modeling framework of risk-sensitive learning. Beyond this setting, we also analyze the single-population BRM in cases where payoff dominance and maximin \emph{differ}, and thus a super-dominant equilibrium does not exist; 
analyze the logit choice rule and its \emph{robustness}, which was not covered in \cite{nax2019risk};  and extend the analyses to asymmetric two-population games with heterogeneous risk parameters.

Furthermore, we show that the same mechanism extends beyond the two-action case. In single-population symmetric $k$-action games under risk-sensitive BRM, if a strongly payoff-dominant equilibrium exists, then it is uniquely selected by sufficiently risk-seeking agents for all sufficiently large populations; similarly, if a strongly maximin equilibrium exists, then it is uniquely selected by sufficiently risk-averse agents for all sufficiently large populations. Thus, 
the risk attitudes may provide a systematic way to steer long-run equilibrium selection for evolutionary dynamics: optimism toward high payoff realizations (i.e., risk-seeking dynamics) may promote payoff-dominant outcome, 
while pessimism toward low payoff realizations (i.e., risk-averse dynamics) may promote maximin outcome and thus robust behaviors. 

\section{Preliminaries}

A \(2\times2\) bimatrix game is specified by the payoff matrices
\[
M^p=
\begin{pmatrix}
a^p & b^p \\
c^p & d^p
\end{pmatrix},
\qquad p\in\{1,2\},
\]
where \(M^p\) denotes the payoff matrix of agent \(p\), and both agents have action set \(\{1,2\}\). In each \(M^p\), rows index agent \(p\)'s own action and columns index the opponent's action. The game is \emph{symmetric} if agents are interchangeable, i.e.,
\[
M^1=M^2=
\begin{pmatrix}
a & b \\
c & d
\end{pmatrix}.
\]

\begin{definition}[\(2\times2\) coordination games]
A \(2\times2\) bimatrix game is a \emph{coordination game} if
\[
a^p>c^p,
\qquad
d^p>b^p,
\qquad p\in\{1,2\}.
\]
In the symmetric setting, these conditions reduce to \(a>c\) and \(d>b\).
\end{definition}

Under these conditions, the pure action profiles \((1,1)\) and \((2,2)\) are \emph{strict Nash equilibria} of the game. We refer to them as \emph{coordinated equilibria}. Unless stated otherwise, we focus on such \(2\times2\) coordination games. We further introduce the following standard refinements of such Nash equilibria.

\begin{definition}[Risk dominance]
The coordinated equilibrium \((1,1)\) is \emph{risk-dominant} \cite{harsanyi1988general} if
\[
(a^1-c^1)(a^2-c^2)>(d^1-b^1)(d^2-b^2).
\]
In the symmetric setting, this condition reduces to \(a-c>d-b\). The coordinated equilibrium \((2,2)\) is risk-dominant if the reverse inequality holds.
\end{definition}

\begin{definition}[Payoff dominance]
The coordinated equilibrium \((1,1)\)  is \emph{payoff-dominant} \cite{harsanyi1988general} if it Pareto-dominates the other coordinated equilibrium, i.e.,
\[
a^p>d^p,
\qquad p\in\{1,2\}.
\]
It is \emph{weakly payoff-dominant} if the above inequalities are non-strict, i.e., \(a^p\ge d^p\) for all \(p\in\{1,2\}\). In the symmetric setting, payoff dominance reduces to \(a>d\), and weak payoff dominance reduces to \(a\ge d\). The coordinated equilibrium \((2,2)\) is payoff-dominant, respectively weakly payoff-dominant, if the reverse strict, respectively non-strict, inequalities hold.
\end{definition}

\begin{definition}[Maximin action and maximin equilibrium]
For agent \(p\), a \emph{maximin action} is an action that maximizes the agent's worst-case payoff over the opponent's actions. Thus, action \(1\) is \emph{weakly maximin} for agent \(p\) if
\[
\min\{a^p,b^p\}\ge \min\{c^p,d^p\}.
\]
It is \emph{strictly maximin} if the above inequality is strict. Similarly, action \(2\) is weakly maximin, respectively strictly maximin, for agent \(p\) if the reverse non-strict, respectively strict, inequality holds. A coordinated equilibrium is weakly maximin if each agent's equilibrium action is weakly maximin.
\end{definition}

\begin{definition}[Super-dominant equilibrium]
A coordinated equilibrium is \emph{super-dominant} \cite{sawa2018reference} if it is weakly maximin for both agents and weakly payoff-dominant, with at least one of these conditions holding strictly. Equivalently, \((1,1)\) is super-dominant if for each \(p\in\{1,2\}\),
\[
a^p\ge d^p,
\qquad
b^p\ge c^p,
\]
and at least one of these inequalities is strict for each $p$. The coordinated equilibrium \((2,2)\) is super-dominant if the reverse inequalities hold, with at least one strict inequality.
\end{definition}

We then consider the standard evolutionary game setting as in \cite{sandholm2010population}, with either one or two populations, i.e., $
\mathcal P\in\bigl\{\{1\},\{1,2\}\bigr\}.$ 
Each population \(p\in\mathcal P\) contains \(N^p\in\mathbb N\) agents, and each agent chooses an action from the set \(S^p=\{1,2\}\). Let \(x^p\in X^{p,N^p}\) denote the fraction of population \(p\) playing action \(1\), where
\[
X^{p,N^p}:=\bigl\{0,\tfrac1{N^p},\ldots,1\bigr\}.
\]
The aggregate state is
\[
x=(x^p)_{p\in\mathcal P}\in X^N
:=\prod_{p\in\mathcal P}X^{p,N^p}.
\]
In the \emph{single-population} setting, the state is one-dimensional. In this case, we write \(N:=N^1\), \(X^N:=X^{1,N}\), and \(x:=x^1\), and suppress the population index throughout for notational simplicity. In the \emph{two-population} setting, we write
\[
N:=(N^1,N^2),
\qquad
X^N \equiv X^{N^1,N^2}:=X^{1,N^1}\times X^{2,N^2}.
\]
Payoffs used in evolutionary learning are given by \(F^p:X^N\to\mathbb R^2\), where 
\[F^p(x)=(F_1^p(x),F_2^p(x)),\] 
and \(F_i^p(x)\) is the payoff to action \(i\in S^p\) for population \(p\) at state \(x\). Agents are repeatedly and randomly matched to play the coordination game. In the single-population setting, matching is with replacement, so the opponent is sampled from the same population according to the current state, and thus \(x^{-p}=x\). In the two-population setting, agents in population \(p\) are matched against population \(-p\), and payoffs depend on the opponent state \(x^{-p}\), where \(-p=2\) if \(p=1\) and \(-p=1\) if \(p=2\).

\section{Risk-Sensitive Evolutionary  Dynamics \& Stochastic Stability}\label{sec:risk-payoffs-dynamics}

\subsection{Risk-Sensitive Evolutionary Dynamics}

To model uncertainty in agents' learning process, we consider risk-sensitive agents who evaluate payoffs via the entropic risk measure  \cite{HowardMatheson1972,Whittle1981,follmer2002convex}, rather than risk-neutral expected payoffs. 

\begin{definition}[Entropic Risk Measure (ERM)]
\label{def:erm}
For a real-valued random payoff $R$ and parameter $\beta\in\mathbb{R}$, the entropic risk measure is defined as
\[
\mathcal{R}_\beta(R) 
:= \frac{1}{\beta}\log \mathbb{E}\!\left[e^{\beta R}\right],
\]
with the continuous extension $\mathcal{R}_0(R)=\mathbb{E}[R]$.
\end{definition}

For \(\beta\) near zero, the Taylor expansion of the ERM gives
\[
\mathcal{R}_{\beta}(R)
=
\mathbb{E}[R]
+\frac{\beta}{2}\operatorname{Var}(R)
+\frac{\beta^{2}}{6}\kappa_{3}(R)
+O(\beta^{3}),
\]
where \(\kappa_{3}(R)=\mathbb{E}[(R-\mathbb{E}[R])^{3}]\). Thus, the ERM
departs from expected payoff by accounting first for variance and then for
skewness. A positive \(\beta\) tilts evaluations toward favorable payoff
realizations and corresponds to risk-seeking behavior, whereas a negative
\(\beta\) emphasizes adverse realizations and corresponds to risk-averse
behavior. In our model, agents within each population \(p\) share a common
risk-sensitivity parameter \(\beta^p\in\mathbb{R}\), although this parameter
may differ across populations.

Conditioning on the population state \(x\), the uncertainty faced by a revising agent comes from \emph{random matching}. Hence, for the \(2\times 2\) bimatrix game, for each \(p\in\mathcal P\), the \emph{ERM-adjusted payoffs} to actions \(1\) and \(2\) are
\begin{align}\label{eq:erm-payoffs}
F_{\beta^p,1}^{p}(x)
&=\frac{1}{\beta^p}\log\!\Big(x^{-p}e^{\beta^p a^p}+(1-x^{-p})e^{\beta^p b^p}\Big),\\
F_{\beta^p,2}^{p}(x)
&=\frac{1}{\beta^p}\log\!\Big(x^{-p}e^{\beta^p c^p}+(1-x^{-p})e^{\beta^p d^p}\Big),\nonumber
\end{align}
with the continuous extension at \(\beta^p=0\) given by the expected payoffs
\begin{align*}
F_{0,1}^{p}(x)
=a^p x^{-p}+b^p(1-x^{-p}),\qquad\qquad 
F_{0,2}^{p}(x)
=c^p x^{-p}+d^p(1-x^{-p}).
\end{align*}
For each fixed \(\beta^p\), these functions are continuous in \(x\) and smooth on the interior. We thus define the \emph{ERM payoff difference} for population \(p\) as
\begin{align}\label{eq:erm-delta}
\Delta_{\beta^p}^{p}(x):=F_{\beta^p,1}^{p}(x)-F_{\beta^p,2}^{p}(x),
\end{align}
which further leads to the following definition. 

\begin{definition}[ERM-adjusted best response]
For population \(p\), an \emph{ERM-adjusted best response} at state \(x\) is an action that maximizes the ERM-adjusted payoff against the population state. Equivalently, action \(1\) is an ERM-adjusted best response if
\[
\Delta_{\beta^p}^p(x)\ge 0,
\]
and action \(2\) is an ERM-adjusted best response if
\[
\Delta_{\beta^p}^p(x)\le 0.
\]
When \(\Delta_{\beta^p}^p(x)=0\), both actions are ERM-adjusted best responses.
\end{definition}

We now introduce our \emph{risk-sensitive evolutionary dynamics}. Note that our focus is on the effect of risk-sensitivity on the \emph{learning process/dynamics}, rather than that on the \emph{equilibrium concept} of the underlying game (as in e.g., \cite{fiat2010players,mazumdar2024tractable}). At each time step \(k\in\{1,2,\ldots\}\), a uniformly random agent receives a revision opportunity and updates its action according to the following risk-sensitive noisy best-response protocol:
\begin{align}\label{eq:risk-rev}
\rho_{1,2}^p(x)
=\sigma^\eta\!\bigl(-\Delta_{\beta^p}^p(x^{-p})\bigr),\qquad\quad 
\rho_{2,1}^p(x)
=\sigma^\eta\!\bigl(\Delta_{\beta^p}^p(x^{-p})\bigr), 
\end{align}
with the convention that, in the single-population setting, we drop the superscript \(p\). Here, \(\sigma^\eta:\mathbb R\to(0,1)\) is parameterized by a noise level \(\eta>0\), and \(\rho_{i,j}^p(x)\) denotes the probability that an agent in population \(p\) switches from action \(i\) to action \(j\) at state \(x\). As the noise level vanishes,
\[
\lim_{\eta\to 0}\sigma^{\eta}(a)=
\begin{cases}
1,& a>0,\\
0,& a<0.
\end{cases}
\]
Hence, \(\sigma^{\eta}(\cdot)\) converges to the ERM-adjusted best response in the vanishing-noise limit. We call the dynamics generated by \eqref{eq:risk-rev} the \emph{risk-sensitive dynamics} for short. To characterize how the probability of choosing a suboptimal action decays as the noise vanishes, we impose the following assumption.

\begin{assumption}\label{ass:cost}
Let \(\{\sigma^\eta\}_{\eta>0}\) be a family of noisy best-response choice functions, with \(\sigma^\eta:\mathbb R\to(0,1)\). For each \(d\in\mathbb R\), assume that
\[
c(d):=-\lim_{\eta\to 0}\eta\log \sigma^\eta(-d)
\]
exists and is finite. The induced function \(c:\mathbb R\to[0,\infty)\) satisfies
\[
c(d)=0 \quad \text{for } d\le0,
\qquad
c(d)>0 \quad \text{for } d>0,
\]
and is nondecreasing on \(\mathbb R\). Moreover, 
\(
\lim_{\eta\to 0}\sigma^\eta(0)=q\in(0,1).
\)
\end{assumption}

We consider two classical noisy best-response protocols from evolutionary game theory: \emph{best response with mutations} \cite{kandori1993learning,young1993evolution} and \emph{logit choice} \cite{blume1993statistical}.

\subsubsection{Best Response with Mutations (BRM)}

Under BRM, a strict ERM-adjusted best response is chosen with probability \(1-e^{-1/\eta}\), and the alternative action is chosen with probability \(e^{-1/\eta}\); ties are broken uniformly:
\[
\sigma^\eta(a)=
\begin{cases}
1-\exp(-1/\eta),& a>0,\\
\frac12,& a=0,\\
\exp(-1/\eta),& a<0.
\end{cases}
\]
This choice rule satisfies Assumption~\ref{ass:cost} with \(c(d)=\mathbf 1_{\{d>0\}}\). We refer to the resulting risk-sensitive dynamics as \texttt{RS-BRM} and \texttt{2Pop-RS-BRM} in the single- and two-population settings, respectively.

\subsubsection{Logit Choice}

Under the logit choice rule,
\[
\sigma^\eta(a)=\frac{e^{a/\eta}}{1+e^{a/\eta}},
\qquad \eta>0,
\]
such that the actions with higher ERM-adjusted payoffs are chosen with higher probability, and \(\sigma^\eta\) converges to the ERM-adjusted best response as \(\eta\to0\). This protocol satisfies Assumption~\ref{ass:cost} with \(c(d)=(d)_+ := \max\{0,d\}\). We refer to the resulting risk-sensitive dynamics as \texttt{RS-logit} and \texttt{2Pop-RS-logit} in the single- and two-population settings, respectively.

For both of these protocols, \(
\lim_{\eta\to 0}\sigma^\eta(0)=\frac{1}{2}
\). We conclude this section with two lemmas that will be used repeatedly later, whose proofs are given in Appendices~\ref{app:lem-delta} and ~\ref{app:lem-xstar}, respectively.

\begin{lemma}\label{lm:delta}
For each population \(p\), the function \(z\mapsto \Delta_{\beta^p}^p(z)\) in \eqref{eq:erm-delta} is strictly increasing on \((0,1)\) and has a unique zero \(x_{\beta^p}^{-p,*}\in(0,1)\). For \(\beta^p\neq 0\), this zero is given by
\begin{equation}\label{eq:xstar-beta}
x_{\beta^p}^{-p,*}
=
\frac{e^{\beta^p d^p}-e^{\beta^p b^p}}
{\bigl(e^{\beta^p a^p}-e^{\beta^p c^p}\bigr)+\bigl(e^{\beta^p d^p}-e^{\beta^p b^p}\bigr)} .
\end{equation}
At \(\beta^p=0\), it has the continuous extension
\begin{align*}
x_{0}^{-p,*}
=
\frac{d^p-b^p}{(a^p-c^p)+(d^p-b^p)}.
\end{align*}
Moreover, \(\Delta_{\beta^p}^p(z)<0\) if and only if \(z<x_{\beta^p}^{-p,*}\), and \(\Delta_{\beta^p}^p(z)>0\) if and only if \(z>x_{\beta^p}^{-p,*}\). Thus, the ERM-adjusted best response changes from action~\(2\) to action~\(1\) as \(z\) crosses \(x_{\beta^p}^{-p,*}\), with both actions being best responses at \(z=x_{\beta^p}^{-p,*}\).
\end{lemma}

\begin{lemma}\label{lem:xstar-beta-monotonicity}
Fix a population \(p\), and let \(x_{\beta^p}^{-p,*}\) be the unique threshold obtained by solving \eqref{eq:xstar-beta}. Then, the map \(\beta^p \mapsto x_{\beta^p}^{-p,*}\) is continuous on \(\mathbb R\). Moreover:
\begin{enumerate}
    \item[(I)] If \(b^p>c^p\) and \(d^p>a^p\), then \(\beta^p \mapsto x_{\beta^p}^{-p,*}\) is strictly increasing, with
    \[
    \lim_{\beta^p\to-\infty} x_{\beta^p}^{-p,*}=0,
    \qquad
    \lim_{\beta^p\to\infty} x_{\beta^p}^{-p,*}=1.
    \]
    
    \item[(II)] If \(b^p<c^p\) and \(d^p<a^p\), then \(\beta^p \mapsto x_{\beta^p}^{-p,*}\) is strictly decreasing, with
    \[
    \lim_{\beta^p\to-\infty} x_{\beta^p}^{-p,*}=1,
    \qquad
    \lim_{\beta^p\to\infty} x_{\beta^p}^{-p,*}=0.
    \]
    
    \item[(III)] If \(b^p\ge c^p\) and \(d^p\le a^p\), with at least one strict inequality, then
    \[
    x_{\beta^p}^{-p,*}<\tfrac12
    \qquad
    \text{for all } \beta^p\in\mathbb R.
    \]
    
    \item[(IV)] If \(b^p\le c^p\) and \(d^p\ge a^p\), with at least one strict inequality, then
    \[
    x_{\beta^p}^{-p,*}>\tfrac12
    \qquad
    \text{for all } \beta^p\in\mathbb R.
    \]
\end{enumerate}
In the degenerate case where \(a^p=d^p\) and \(b^p=c^p\), we have \(x_{\beta^p}^{-p,*}\equiv \tfrac12\). In the single-population setting, the superscripts \(p\) and \(-p\) are omitted for notational convenience. 
\end{lemma}

\subsection{A Unified Framework for Stochastic-Stability Analysis}\label{subsec:stochastic-stability}

We now introduce a stochastic-stability analysis framework, adapted from \cite{sandholm2010orders,staudigl2012stochastic}, that will unify our subsequent analyses in the single-population and two-population settings. Throughout this subsection, \(N\) is fixed. For each \(\eta>0\), the risk-sensitive dynamics induce a discrete-time Markov chain
\[
X^{N,\eta}:=\{X_k^{N,\eta}\}_{k\ge 0}
\]
on the finite state space \(X^N\), where \(X_k^{N,\eta}\) denotes the population state after \(k\) revision opportunities. For the noisy best-response protocols studied here, each action is chosen with strictly positive probability for every \(\eta>0\). Hence, from any state, every adjacent state can be reached with positive probability whenever the corresponding transition is feasible, and the chain also has positive self-loop probability. Therefore, the chain is irreducible and aperiodic on \(X^N\). Since \(X^N\) is finite, the standard finite-state Markov-chain theorem implies that the chain admits a unique stationary distribution, denoted by \(\mu^{N,\eta}\); see e.g., \cite{stroock2005introduction} for the standard argument on this result.


\begin{definition}[Stochastic stability]\label{def:sse}
A state \(x\in X^N\) is \emph{stochastically stable} if 
\[\lim_{\eta\to0}\mu^{N,\eta}(x) >0,\]
and \emph{uniquely stochastically stable} if 
\[\lim_{\eta\to0}\mu^{N,\eta}(x)=1.\]
\end{definition}

Intuitively, a stochastically stable state retains strictly positive stationary probability as the noise level vanishes. Unique stochastic stability means that the stationary distribution concentrates entirely on that state. This is the standard stochastic-stability criterion for finite-state small-noise Markov chains; see e.g., \cite{young1993evolution,kandori1993learning}.

To analyze stochastic stability, we record the exponential decay rates of one-step transition probabilities. For states \(x,y\in X^N\), let
\[
P^{N,\eta}(x,y)
:=
\mathbb P\!\left(X_{k+1}^{N,\eta}=y \,\middle|\, X_k^{N,\eta}=x\right)
\]
denote the one-step transition probability. Since the noisy best-response choice functions satisfy \(\sigma^\eta(a)\in(0,1)\) for every \(a\in\mathbb R\) and every \(\eta>0\), the support of \(P^{N,\eta}\) is independent of \(\eta\). Thus, if a one-step transition has positive probability for some \(\eta>0\), then it has positive probability for every \(\eta>0\).

\begin{definition}[Transition costs]
For distinct states \(x\neq y\), we write \(x\to y\) if \(y\) is reachable from \(x\) in one feasible revision step, which is equivalent to \(P^{N,\eta}(x,y)>0\) for every \(\eta>0\). Under Assumption~\ref{ass:cost}, the \emph{small-noise cost} of a feasible one-step transition \(x\to y\) is
\begin{equation}\label{eq:ldp-step}
\kappa^N(x,y):=
\lim_{\eta\to 0}\bigl[-\eta\log P^{N,\eta}(x,y)\bigr].
\end{equation}
For infeasible transitions, we set \(\kappa^N(x,y)=+\infty\). A feasible transition \(x\to y\) is called a \emph{zero-cost transition} if \(\kappa^N(x,y)=0\).
\end{definition}

The limit in \eqref{eq:ldp-step} exists because the non-exponential factors in \(P^{N,\eta}(x,y)\), such as the probability of selecting a particular population or action type, do not affect the value of \(-\eta\log P^{N,\eta}(x,y)\) as \(\eta\to0\), while Assumption~\ref{ass:cost} controls the exponential decay rate of the revision probability. Under the revision protocols considered here, zero-cost transitions are exactly the one-step updates that either follow an ERM-adjusted best response or occur along a tie.

\begin{definition}[Paths and path costs]
A \emph{path} is a finite sequence of states $\phi=(\phi_0,\phi_1,\ldots,\phi_K)$ 
in \(X^N\) such that \(\phi_k\to\phi_{k+1}\) for all \(k=0,\ldots,K-1\). The \emph{cost} of a path \(\phi\) is
\[
C^N(\phi):=\sum_{k=0}^{K-1}\kappa^N(\phi_k,\phi_{k+1}).
\]
For states \(x,y\in X^N\), the \emph{minimal path cost} from \(x\) to \(y\) is
\begin{align*}
C^N(x,y):=
\min\bigl\{C^N(\phi):\phi_0=x,\ \phi_K=y\bigr\}.
\end{align*}
More generally, for a state \(x\in X^N\) and a nonempty set \(E\subseteq X^N\), the \emph{minimal path cost} from \(x\) to \(E\) is
\begin{align*}
C^N(x,E):=
\min\bigl\{C^N(\phi):\phi_0=x,\ \phi_K\in E\bigr\}.
\end{align*}
\end{definition}

The zero-cost transitions induce a directed graph on \(X^N\). We refer to this graph as the \emph{zero-noise dynamics}, because it describes the limiting motion of the noisy best response dynamics as \(\eta\to0\).

\begin{definition}[Attracting state, stability set, and exit boundary]\label{def:zero-noise-stability-set}
A state \(A\in X^N\) is an
\emph{attracting state} if \(\{A\}\) is a singleton sink recurrent class of the zero-noise dynamics; equivalently, there is no zero-cost transition from \(A\) to any
state outside \(\{A\}\). 
The \emph{stability set} of \(A\), denoted by \(\mathcal B_0^N(A)\), is the largest subset of \(X^N\) containing \(A\) such that every state in \(\mathcal B_0^N(A)\) reaches \(A\) along a zero-cost path, and no state in \(\mathcal B_0^N(A)\) has a zero-cost path to \(X^N\setminus \mathcal B_0^N(A)\). The corresponding \emph{exit boundary} of \(A\) is
\[
\Exit^N(A):=
\Bigl\{
x\in X^N\setminus \mathcal B_0^N(A):
\exists\, y\in \mathcal B_0^N(A)\ \text{such that } y\to x
\Bigr\}.
\]
\end{definition}

In coordination games, the zero-noise dynamics have two attracting states corresponding to the two coordinated equilibria. The following assumption isolates this structure.

\begin{assumption}\label{ass:2attractor}
The zero-noise dynamics have exactly two attracting states, denoted by \(A\) and \(B\). Moreover:
\begin{enumerate}
\item[(I)] Every state \(x\in X^N\setminus\{A,B\}\) has a zero-cost path to at least one of \(A\) or \(B\).
\item[(II)] If \(x\notin \mathcal B_0^N(A)\), then \(C^N(x,B)=0\), and if \(x\notin \mathcal B_0^N(B)\), then \(C^N(x,A)=0\).
\end{enumerate}
\end{assumption}

Assumption~\ref{ass:2attractor} will be verified explicitly in both the single-population and two-population settings. Under this assumption, stochastic-stability analysis reduces to comparing the minimal costs of escaping the two stability sets, as formalized in the following theorem.

\begin{theorem}[Stochastic stability via escape costs]\label{thm:sse}
Fix the population-size parameter \(N\), and consider the Markov chain \(X^{N,\eta}\) induced by the risk-sensitive dynamics. Suppose Assumptions~\ref{ass:cost} and~\ref{ass:2attractor} hold, with attracting states \(A\) and \(B\). Define
\[
\gamma_A^N:=C^N\!\bigl(B,\Exit^N(B)\bigr),
\qquad
\gamma_B^N:=C^N\!\bigl(A,\Exit^N(A)\bigr).
\]
Then, we have
\begin{equation}\label{eq:mass-ratio}
\lim_{\eta\to 0}\eta\log\frac{\mu^{N,\eta}(B)}{\mu^{N,\eta}(A)}
=
\gamma_A^N-\gamma_B^N.
\end{equation}
Consequently, \(A\) is uniquely stochastically stable if \(\gamma_A^N<\gamma_B^N\), while \(B\) is uniquely stochastically stable if \(\gamma_B^N<\gamma_A^N\).
\end{theorem}

Theorem~\ref{thm:sse} follows from the standard Freidlin--Wentzell tree representation used in stochastic-stability analysis; see e.g., \cite{freidlin1998random,staudigl2012stochastic,sandholm2016large}. For completeness, we provide a short proof in Appendix~\ref{app:thm-sse-proof}.

\section{Single-Population Symmetric Setting}

Since only one agent revises at each time step, the feasible one-step transitions are \(x\to x+\frac1N\), generated by an agent switching from action \(2\) to action \(1\), and \(x\to x-\frac1N\), generated by an agent switching from action \(1\) to action \(2\). Recalling the one-step cost function \(c(\cdot)\) from Assumption~\ref{ass:cost}, \eqref{eq:ldp-step} gives
\begin{align}\label{eq:1pop-kappa-compact}
\kappa^N\!\left(x,x+\tfrac1N\right)
=c\!\bigl((-\Delta_\beta(x))_+\bigr),\qquad\qquad  
\kappa^N\!\left(x,x-\tfrac1N\right)
=c\!\bigl((\Delta_\beta(x))_+\bigr).
\end{align}

Let \(x_\beta^*\in(0,1)\) be the unique threshold from Lemma~\ref{lm:delta}. Define the adjacent grid points
\[
\underline{x}_\beta^N:=\frac{\lfloor N x_\beta^*\rfloor}{N},
\qquad
\bar{x}_\beta^N:=\frac{\lceil N x_\beta^*\rceil}{N}.
\]
The next proposition verifies Assumption~\ref{ass:2attractor}.

\begin{proposition}\label{prop:1pop-ass2}
Under the single-population risk-sensitive dynamics induced by \eqref{eq:risk-rev}, with transition costs given by \eqref{eq:1pop-kappa-compact}, the zero-noise dynamics have exactly two attracting states, namely \(0\) and \(1\).

If \(N x_\beta^*\notin\mathbb Z\), then
\[
\mathcal B_0^N(0)
=
\Bigl\{0,\frac1N,\ldots,\underline{x}_\beta^N\Bigr\},
\qquad
\mathcal B_0^N(1)
=
\Bigl\{\bar{x}_\beta^N,\bar{x}_\beta^N+\frac1N,\ldots,1\Bigr\},
\]
with
\begin{align*}
\Exit^N(0)=\{\bar{x}_\beta^N\},
\qquad
\Exit^N(1)=\{\underline{x}_\beta^N\}.
\end{align*}
If \(N x_\beta^*\in\mathbb Z\), so that \(x_\beta^*\in X^N\), then
\[
\mathcal B_0^N(0)
=
\Bigl\{0,\frac1N,\ldots,x_\beta^*-\frac1N\Bigr\},
\qquad
\mathcal B_0^N(1)
=
\Bigl\{x_\beta^*+\frac1N,x_\beta^*+\frac2N,\ldots,1\Bigr\},
\]
with
\begin{align*}
\Exit^N(0)=\Exit^N(1)=\{x_\beta^*\}.
\end{align*}
In either case, Assumption~\ref{ass:2attractor} holds with \(A=0\) and \(B=1\).
\end{proposition}

\begin{proof}
By Lemma~\ref{lm:delta}, \(\Delta_\beta(x)<0\) for \(x<x_\beta^*\), and \(\Delta_\beta(x)>0\) for \(x>x_\beta^*\). Hence, by \eqref{eq:1pop-kappa-compact}, every state \(x<x_\beta^*\) has a zero-cost downward move \(x\to x-\frac1N\), while every state \(x>x_\beta^*\) has a zero-cost upward move \(x\to x+\frac1N\). If \(x_\beta^*\in X^N\), then \(\Delta_\beta(x_\beta^*)=0\), so both moves
\[
x_\beta^*\to x_\beta^*-\frac1N,
\qquad
x_\beta^*\to x_\beta^*+\frac1N
\]
have zero cost.

Define
\[
E_0:=\{x\in X^N:x<x_\beta^*\},
\qquad
E_1:=\{x\in X^N:x>x_\beta^*\}.
\]
If \(N x_\beta^*\notin\mathbb Z\), then
\[
E_0=\Bigl\{0,\frac1N,\ldots,\underline{x}_\beta^N\Bigr\},
\qquad
E_1=\Bigl\{\bar{x}_\beta^N,\bar{x}_\beta^N+\frac1N,\ldots,1\Bigr\}.
\]
If \(N x_\beta^*\in\mathbb Z\), then
\[
E_0=\Bigl\{0,\frac1N,\ldots,x_\beta^*-\frac1N\Bigr\},
\qquad
E_1=\Bigl\{x_\beta^*+\frac1N,x_\beta^*+\frac2N,\ldots,1\Bigr\}.
\]

Every state in \(E_0\) reaches \(0\) along repeated downward zero-cost moves. Moreover, no state in \(E_0\) has a zero-cost path to \(X^N\setminus E_0\), since any such path must contain an upward move from some state \(x<x_\beta^*\), which has positive cost. Therefore, \(\mathcal B_0^N(0)=E_0\). Similarly, every state in \(E_1\) reaches \(1\) along repeated upward zero-cost moves, and no state in \(E_1\) has a zero-cost path to \(X^N\setminus E_1\). Hence, \(\mathcal B_0^N(1)=E_1\).

The formulas for the exit boundaries follow directly. If \(N x_\beta^*\notin\mathbb Z\), the only one-step transition from \(\mathcal B_0^N(0)\) to its complement is  $
\underline{x}_\beta^N \to \bar{x}_\beta^N$, 
so \(\Exit^N(0)=\{\bar{x}_\beta^N\}\). Similarly, the only one-step transition from \(\mathcal B_0^N(1)\) to its complement is $\bar{x}_\beta^N \to \underline{x}_\beta^N$, 
so \(\Exit^N(1)=\{\underline{x}_\beta^N\}\). If \(N x_\beta^*\in\mathbb Z\), then the first state outside either stability set that can be reached in one step is \(x_\beta^*\), and thus
\[
\Exit^N(0)=\Exit^N(1)=\{x_\beta^*\}.
\]

It remains to verify Assumption~\ref{ass:2attractor}. The zero-noise dynamics move downward on \(E_0\) and upward on \(E_1\), while at \(x_\beta^*\), when \(x_\beta^*\in X^N\), both adjacent moves have zero cost. Hence, the only attracting states are \(0\) and \(1\). Assumption~\ref{ass:2attractor} \textup{(I)} holds because every state \(x\in X^N\setminus\{0,1\}\) either lies in \(E_0\), in which case it has a zero-cost path to \(0\), or lies in \(E_1\), in which case it has a zero-cost path to \(1\), or equals \(x_\beta^*\), in which case it has zero-cost paths to both \(0\) and \(1\). Assumption~\ref{ass:2attractor} \textup{(II)} holds because if \(x\notin \mathcal B_0^N(0)\), then \(x\ge x_\beta^*\), so \(x\) has a zero-cost path to \(1\); similarly, if \(x\notin \mathcal B_0^N(1)\), then \(x\le x_\beta^*\), so \(x\) has a zero-cost path to \(0\).
\end{proof}

With the above notation, the escape costs from the two stability sets are
\begin{align}\label{eq:escape-costs-1pop-general}
\gamma_0^N
:=C^N\!\bigl(1,\Exit^N(1)\bigr)
=
C^N\!\bigl(1,\{\underline{x}_\beta^N\}\bigr),~~~
\gamma_1^N
:=C^N\!\bigl(0,\Exit^N(0)\bigr)
=
C^N\!\bigl(0,\{\bar{x}_\beta^N\}\bigr).
\end{align}
We also define the selection coefficient
\begin{equation}\label{eq:SN}
S^N:=\gamma_1^N-\gamma_0^N.
\end{equation}
By Theorem~\ref{thm:sse}, state \(1\) is uniquely stochastically stable if \(S^N<0\), whereas state \(0\) is uniquely stochastically stable if \(S^N>0\). We next establish the stochastic stability results under both risk-sensitive dynamics of \texttt{RS-BRM} and \texttt{RS-logit}.

\begin{theorem}[Stochastic stability under \texttt{RS-BRM}]
\label{thm:rsbrm_sse}
Fix \(\beta\in\mathbb R\), and let \(x_\beta^*\in(0,1)\) be the unique threshold from Lemma~\ref{lm:delta}. Then, there exists \(N_0(\beta)\) such that, for every \(N\ge N_0(\beta)\), the unique stochastically stable state under \texttt{RS-BRM} is $1$ if $x_\beta^*<\tfrac12$, and is $0$ if $x_\beta^*>\tfrac12$. Consequently, for every fixed \(\beta\), the following conclusions hold
for all \(N\ge N_0(\beta)\):
\begin{enumerate}
\item[(I)] If \(a>d\) and \(b<c\), or \(a<d\) and \(b>c\), then there
    exists a unique \(\beta^\dagger\in\mathbb R\) such that the
    payoff-dominant equilibrium is uniquely stochastically stable when
    \(\beta>\beta^\dagger\), whereas the maximin equilibrium is uniquely
    stochastically stable when \(\beta<\beta^\dagger\).

\item[(II)] If \(a\ge d\) and \(b\ge c\), or \(a\le d\) and \(b\le c\),
    with at least one strict inequality, then the super-dominant equilibrium
    is uniquely stochastically stable.
\end{enumerate}
\end{theorem}

\begin{proof}
By Proposition~\ref{prop:1pop-ass2}, Assumption~\ref{ass:2attractor} holds with attracting states \(0\) and \(1\). Under \texttt{RS-BRM}, the cost function is \(c(d)=\mathbf 1_{\{d>0\}}\). Hence, a transition that moves opposite to the unique ERM-adjusted best response has cost \(1\), whereas a transition that follows an ERM-adjusted best response, or occurs at a tie, has cost \(0\).

Starting from \(1\), any path to \(\underline{x}_\beta^N=\lfloor N x_\beta^*\rfloor/N\) must contain at least $N-\lfloor N x_\beta^*\rfloor$ 
 downward steps. The monotone path
\[
\left(1,\ 1-\frac1N,\ \ldots,\ \underline{x}_\beta^N\right)
\]
attains this lower bound. Each downward step on this path starts from a state strictly larger than \(x_\beta^*\). Indeed, the last such step starts from \(\underline{x}_\beta^N+\frac1N>x_\beta^*\) if \(N x_\beta^*\notin\mathbb Z\), and from \(x_\beta^*+\frac1N>x_\beta^*\) if \(N x_\beta^*\in\mathbb Z\). At such states, action \(1\) is the unique ERM-adjusted best response, so a downward step, which decreases the share of action \(1\), moves opposite to the ERM-adjusted best-response direction and has cost \(1\). Therefore, we have 
\[
\gamma_0^N=N-\lfloor N x_\beta^*\rfloor.
\]

Similarly, any path from \(0\) to \(\bar{x}_\beta^N=\lceil N x_\beta^*\rceil/N\) must contain at least $
\lceil N x_\beta^*\rceil$ 
upward steps. The monotone path
\[
\left(0,\ \frac1N,\ \ldots,\ \bar{x}_\beta^N\right)
\]
attains this lower bound. Each upward step on this path starts from a state strictly smaller than \(x_\beta^*\). At such states, action \(2\) is the unique ERM-adjusted best response, so an upward step, which increases the share of action \(1\), moves opposite to the ERM-adjusted best-response direction and has cost \(1\). Thus, we have 
\[
\gamma_1^N=\lceil N x_\beta^*\rceil.
\]

By Theorem~\ref{thm:sse}, state \(1\) is uniquely stochastically stable if \(\gamma_1^N<\gamma_0^N\), while state \(0\) is uniquely stochastically stable if \(\gamma_0^N<\gamma_1^N\). Since
\begin{align*}
\frac{\gamma_1^N}{N}
=
\frac{\lceil N x_\beta^*\rceil}{N}
\to x_\beta^*,
\qquad
\frac{\gamma_0^N}{N}
=
1-\frac{\lfloor N x_\beta^*\rfloor}{N}
\to 1-x_\beta^*,
\end{align*}
the sign of \(\gamma_1^N-\gamma_0^N\) agrees with the sign of \(x_\beta^*-\tfrac12\) for all sufficiently large \(N\), provided \(x_\beta^*\neq\tfrac12\). This proves the first statement.

The remaining claims follow from Lemma~\ref{lem:xstar-beta-monotonicity}. Suppose first that \(a<d\) and \(b>c\). Then \((2,2)\), corresponding to state \(0\), is payoff-dominant, while action \(1\) is maximin, so \((1,1)\), corresponding to state \(1\), is the maximin equilibrium. By Lemma~\ref{lem:xstar-beta-monotonicity}, \(\beta\mapsto x_\beta^*\) is continuous and strictly increasing, with limits \(0\) and \(1\) as \(\beta\to-\infty\) and \(\beta\to+\infty\), respectively. Hence, there exists a unique \(\beta^\dagger\) such that \(x_{\beta^\dagger}^*=\tfrac12\). For every fixed \(\beta<\beta^\dagger\), we have \(x_\beta^*<\tfrac12\), so state \(1\), the maximin equilibrium, is stochastically stable for all sufficiently large \(N\). For every fixed \(\beta>\beta^\dagger\), we have \(x_\beta^*>\tfrac12\), so state \(0\), the payoff-dominant equilibrium, is stochastically stable for all sufficiently large \(N\).

Now suppose that \(a>d\) and \(b<c\). Then \((1,1)\), corresponding to state \(1\), is payoff-dominant, while action \(2\) is maximin, so \((2,2)\), corresponding to state \(0\), is the maximin equilibrium. In this case, Lemma~\ref{lem:xstar-beta-monotonicity} implies that \(\beta\mapsto x_\beta^*\) is continuous and strictly decreasing, again crossing \(\tfrac12\) at a unique \(\beta^\dagger\). Thus, for every fixed \(\beta<\beta^\dagger\), we have \(x_\beta^*>\tfrac12\), so state \(0\), the maximin equilibrium, is stochastically stable for all sufficiently large \(N\). For every fixed \(\beta>\beta^\dagger\), we have \(x_\beta^*<\tfrac12\), so state \(1\), the payoff-dominant equilibrium, is stochastically stable for all sufficiently large \(N\). This proves \textup{(I)}.

Finally, suppose that \(a\ge d\) and \(b\ge c\), with at least one strict inequality. Then, by definition,  \((1,1)\) is super-dominant, and Lemma~\ref{lem:xstar-beta-monotonicity} gives \(x_\beta^*<\tfrac12\) for all \(\beta\in\mathbb R\). Hence, state \(1\) is stochastically stable for every fixed \(\beta\), once \(N\) is sufficiently large. The case \(a\le d\) and \(b\le c\), with at least one strict inequality, is symmetric: \((2,2)\) is super-dominant, \(x_\beta^*>\tfrac12\) for all \(\beta\), and state \(0\) is stochastically stable for every fixed \(\beta\), once \(N\) is sufficiently large. This proves \textup{(II)} and completes the proof.
\end{proof}

\begin{theorem}[Stochastic stability under \texttt{RS-logit}]
\label{thm:rslogit_sse}
Define
\begin{align}\label{eq:bar-s-logit}
    \bar S_{\mathrm{logit}}(\beta)
:=
-\int_0^1 \Delta_\beta(x)\,dx.
\end{align}
Fix \(\beta\in\mathbb R\) and assume \(\bar S_{\mathrm{logit}}(\beta)\neq 0\). Then, there exists \(N_0(\beta)\) such that, for every \(N\ge N_0(\beta)\), the unique stochastically stable state under \texttt{RS-logit} is $1$ if $\bar S_{\mathrm{logit}}(\beta)<0$, and $0$ if $ \bar S_{\mathrm{logit}}(\beta)>0$.
Consequently, for every fixed \(\beta\), the following conclusions hold
for all \(N\ge N_0(\beta)\):
\begin{enumerate}
\item[(I)] If \(a>d\) and \(b<c\), or \(a<d\) and \(b>c\), then there exists a unique \(\beta^\dagger\in\mathbb R\) such that the payoff-dominant equilibrium is uniquely stochastically stable for \(\beta > \beta^\dagger\), while the maximin equilibrium is uniquely stochastically stable for \(\beta < \beta^\dagger\).

\item[(II)] If \(a\ge d\) and \(b\ge c\), or \(a\le d\) and \(b\le c\), with at least one strict inequality, then the super-dominant equilibrium is uniquely stochastically stable.
\end{enumerate}
\end{theorem}

\begin{proof}
By Proposition~\ref{prop:1pop-ass2}, Assumption~\ref{ass:2attractor} holds with attracting states \(0\) and \(1\). Under \texttt{RS-logit}, the cost function is \(c(d)=(d)_+\). Therefore, by \eqref{eq:1pop-kappa-compact},
\[
\kappa^N\!\left(x,x+\tfrac1N\right)=(-\Delta_\beta(x))_+,
\qquad
\kappa^N\!\left(x,x-\tfrac1N\right)=(\Delta_\beta(x))_+.
\]

Since the state space is one-dimensional and all step costs are nonnegative, the minimum-cost paths to the exit boundaries are the monotone paths. Hence,
\[
\gamma_0^N
=
\sum_{m=\lfloor N x_\beta^*\rfloor+1}^{N}
\Delta_\beta\!\left(\frac{m}{N}\right),
\qquad
\gamma_1^N
=
\sum_{m=0}^{\lceil N x_\beta^*\rceil-1}
\left[-\Delta_\beta\!\left(\frac{m}{N}\right)\right].
\]
The summands are nonnegative: the first sum runs over states \(x>x_\beta^*\), where \(\Delta_\beta(x)>0\), and the second runs over states \(x<x_\beta^*\), where \(\Delta_\beta(x)<0\).

Recalling the definition \(S^N=\gamma_1^N-\gamma_0^N\) from \eqref{eq:SN}, define
\[
\bar S_{\mathrm{logit}}^{\,N}(\beta):=\frac1N S^N.
\]
The above expressions imply
\begin{equation}\label{eq:S_N_bar_logit}
\bar S_{\mathrm{logit}}^{\,N}(\beta)
=
-\frac1N\sum_{m=0}^{N}
\Delta_\beta\!\left(\frac{m}{N}\right),  
\end{equation}
where, if \(N x_\beta^*\in\mathbb Z\), the threshold term contributes zero. Since \(\Delta_\beta\) is continuous on \([0,1]\), this is a Riemann sum, and therefore
\[
\bar S_{\mathrm{logit}}^{\,N}(\beta)
\longrightarrow
\bar S_{\mathrm{logit}}(\beta)
\qquad
\text{as } N\to\infty.
\]
Thus, whenever \(\bar S_{\mathrm{logit}}(\beta)\neq0\), the sign of \(S^N\) agrees with the sign of \(\bar S_{\mathrm{logit}}(\beta)\) for all sufficiently large \(N\). The first claim then follows from Theorem~\ref{thm:sse}.

It remains to identify the selected equilibrium in the payoff-dominant/maximin and super-dominant cases. By Lemma~\ref{lem:Slogit}, \(\bar S_{\mathrm{logit}}\) is continuous and monotone on \(\mathbb R\), with
\begin{align*}
\lim_{\beta\to-\infty}\bar S_{\mathrm{logit}}(\beta)
&=
-\min\{a,b\}+\min\{c,d\},\\
\lim_{\beta\to+\infty}\bar S_{\mathrm{logit}}(\beta)
&=
-\max\{a,b\}+\max\{c,d\}.
\end{align*}

In fact, when payoff-dominant and maximin equilibria differ, the coordination-game inequalities imply
\(|a-b|\neq |c-d|\): specifically, \(a<d\) and \(b>c\) imply
\(|a-b|<|c-d|\), while \(a>d\) and \(b<c\) imply
\(|a-b|>|c-d|\). Hence, 
\(\bar S_{\mathrm{logit}}\) is strictly monotone in these cases.

Suppose first that \(a<d\) and \(b>c\). Then \((2,2)\), which corresponds  to state \(0\), is payoff-dominant, while \((1,1)\), which corresponds  to state \(1\), is the maximin equilibrium. In this case, the two endpoint limits of \(\bar S_{\mathrm{logit}}\) have opposite signs. Hence, by Lemma~\ref{lem:Slogit}, there exists a unique \(\beta^\dagger\in\mathbb R\) such that
\[
\bar S_{\mathrm{logit}}(\beta^\dagger)=0.
\]
For every fixed \(\beta<\beta^\dagger\), we have \(\bar S_{\mathrm{logit}}(\beta)<0\), so state \(1\), the maximin equilibrium, is stochastically stable for all sufficiently large \(N\). For every fixed \(\beta>\beta^\dagger\), we have \(\bar S_{\mathrm{logit}}(\beta)>0\), so state \(0\), the payoff-dominant equilibrium, is stochastically stable for all sufficiently large \(N\).

Now suppose that \(a>d\) and \(b<c\). Then \((1,1)\), which corresponds  to state \(1\), is payoff-dominant, while \((2,2)\), which corresponds  to state \(0\), is the maximin equilibrium. Again, Lemma~\ref{lem:Slogit} gives a unique \(\beta^\dagger\) at which \(\bar S_{\mathrm{logit}}\) changes sign. For every fixed \(\beta<\beta^\dagger\), state \(0\), the maximin equilibrium, is stochastically stable for all sufficiently large \(N\), while for every fixed \(\beta>\beta^\dagger\), state \(1\), the payoff-dominant equilibrium, is stochastically stable for all sufficiently large \(N\). This proves \textup{(I)}.

Finally, suppose that \(a\ge d\) and \(b\ge c\), with at least one strict inequality. Then, equilibrium \((1,1)\) is super-dominant. In this case,
\[
\min\{a,b\}\ge \min\{c,d\},
\qquad
\max\{a,b\}\ge \max\{c,d\},
\]
and \(\bar S_{\mathrm{logit}}(0)<0\). By Lemma~\ref{lem:Slogit}, it follows that
\[
\bar S_{\mathrm{logit}}(\beta)<0
\qquad
\text{for all } \beta\in\mathbb R.
\]
Hence, state \(1\) is stochastically stable for every fixed \(\beta\), once \(N\) is sufficiently large. The case \(a\le d\) and \(b\le c\), with at least one strict inequality, is symmetric: \((2,2)\) is super-dominant, \(\bar S_{\mathrm{logit}}(\beta)>0\) for all \(\beta\in\mathbb R\), and state \(0\) is stochastically stable for every fixed \(\beta\), once \(N\) is sufficiently large. This proves \textup{(II)}.
\end{proof}

\begin{remark}
    The conclusions of Theorems~\ref{thm:rsbrm_sse} and~\ref{thm:rslogit_sse} are pointwise in \(\beta\): \(N_0(\beta)\) need not be uniformly bounded and may diverge as \(\beta\) approaches the corresponding critical value \(\beta^\dagger\).
\end{remark}




Theorems~\ref{thm:rsbrm_sse} and~\ref{thm:rslogit_sse} show that, despite their differences, \texttt{RS-BRM} and \texttt{RS-logit} induce the same equilibrium-selection pattern in symmetric \(2\times 2\) coordination games. When one equilibrium is super-dominant, long-run equilibrium selection is robust to both agents' risk attitudes and the choice of noisy best-response protocols. Intuitively, a super-dominant equilibrium \emph{aligns} payoff dominance with worst-case robustness: it yields higher coordinated payoffs while also protecting agents against unfavorable mismatches. Thus, when such an equilibrium exists, there is no tension between \emph{efficiency} and \emph{safety}, and the learning dynamics select it regardless of whether agents are risk-averse, risk-neutral, or risk-seeking. This robustness is desirable because the selected outcome does not depend sensitively on the precise risk preferences of the population or on the specific noise model used in the revision protocol.

For example, consider  $M=
\begin{pmatrix}
4 & 1\\
0 & 2
\end{pmatrix}.$ 
Here, \((1,1)\) is payoff-dominant because \(4>2\), and action \(1\) is also maximin because
\[
\min\{4,1\}=1
>
0=\min\{0,2\}.
\]
Hence, \((1,1)\) is super-dominant. Theorems~\ref{thm:rsbrm_sse} and~\ref{thm:rslogit_sse} imply that \((1,1)\) is stochastically stable for every risk-sensitivity parameter \(\beta\), under both \texttt{RS-BRM} and \texttt{RS-logit}, once the population is sufficiently large.

On the other hand, when the payoff-dominant and maximin equilibria \emph{differ}, there is a fundamental tension between efficiency and robustness. In this case, the risk-sensitivity parameter can act  as a knob for equilibrium selection: as agents become more risk-seeking, the payoff-dominant equilibrium, which Pareto-dominates the other coordinated equilibrium and yields higher coordinated payoffs, becomes stochastically stable; whereas as agents become more risk-averse, the maximin equilibrium, which provides better protection against worst-case opponent actions, becomes stochastically stable.

For example, consider the symmetric coordination game
\[
M=
\begin{pmatrix}
4 & 0\\
2 & 3
\end{pmatrix}.
\]
Here, \((1,1)\) is payoff-dominant because \(4>3\), while \((2,2)\) is maximin because
\[
\min\{4,0\}=0
<
2=\min\{2,3\}.
\]
Thus, efficiency and robustness point to different equilibria. Theorems~\ref{thm:rsbrm_sse} and~\ref{thm:rslogit_sse} imply that, under both \texttt{RS-BRM} and \texttt{RS-logit}, sufficiently risk-averse agents select the maximin equilibrium \((2,2)\), whereas sufficiently risk-seeking agents select the payoff-dominant equilibrium \((1,1)\), for all sufficiently large populations. The cutoff value of the risk-sensitivity parameter may differ across the two protocols, but the transition is the same: increasing \(\beta\) shifts long-run selection from the ``safer'' and more robust maximin equilibrium to the more ``efficient'' payoff-dominant equilibrium. These insights align with recent empirical observations in \cite{noorani2022risk,zhang2025optimism} that risk-seeking  learning agents can improve coordination efficiency, and our results may be viewed as a formalization of these insights through the lens of coordination games and evolutionary dynamics. 

\section{Extension: Two-Population Asymmetric Setting}

We now consider a \emph{two-population} setting with an \emph{asymmetric} coordination game. This setting is natural because many coordination problems involve distinct agent roles with different payoffs, such as the Battle of the Sexes \cite{osborne1994course}. Notably, this setting also allows for \emph{heterogeneous risk attitudes}: agents are homogeneous within each population, sharing a common risk-sensitivity parameter \(\beta^p\), while the two populations may have different risk parameters, so \(\beta^1\) and \(\beta^2\) need not be equal.

Moving from one to two populations requires more than a change in notation. In the single-population model, the aggregate state is one-dimensional, and a single risk-dependent threshold separates the two best-response regions, yielding closed-form escape costs. In the two-population model, the state is two-dimensional, and each population’s best response depends on the state of the other. Although BRM remains tractable, transition costs under logit dynamics are state- and path-dependent, requiring a genuinely two-dimensional analysis. We characterize stochastic selection under super-dominance, while the general risk-sensitive two-population logit case remains open.

In the two-population setting, the aggregate population state is
\[
x=(x^1,x^2)\in X^N,
\]
where \(x^p\) is the fraction of population \(p\) playing action \(1\). A one-step update changes the action of only one agent in one population. Let \(e_1=(1,0)\) and \(e_2=(0,1)\). Then, for \(p\in\{1,2\}\), the nontrivial feasible one-step transitions are
\[
x\to x+\frac1{N^p}e_p,
\qquad
x\to x-\frac1{N^p}e_p,
\]
whenever the resulting state remains in \(X^N\). Here, \(x+\frac1{N^p}e_p\) corresponds to an agent in population \(p\) switching from action \(2\) to action \(1\), while \(x-\frac1{N^p}e_p\) corresponds to an agent in population \(p\) switching from action \(1\) to action \(2\).

Applying the small-noise cost definition in \eqref{eq:ldp-step} to the two-population state space, we write the corresponding transition costs as \(\kappa^N\), which may also be denoted as  \(\kappa^{N^1,N^2}\) when we want to emphasize the two population sizes. Recalling the risk-sensitive revision protocol \eqref{eq:risk-rev} and the cost function \(c(\cdot)\) from Assumption~\ref{ass:cost}, we obtain
\begin{align}\label{eq:2pop-kappa}
\kappa^N\!\left(x,x+\frac1{N^p}e_p\right)=
c\!\left((-\Delta_{\beta^p}^p(x^{-p}))_+\right),~ 
\kappa^N\!\left(x,x-\frac1{N^p}e_p\right)=
c\!\left((\Delta_{\beta^p}^p(x^{-p}))_+\right).
\end{align}

Next, let \(x_{\beta^1}^{2,*}\in(0,1)\) and \(x_{\beta^2}^{1,*}\in(0,1)\) be the unique thresholds from Lemma~\ref{lm:delta}, defined by
\[
\Delta_{\beta^1}^{1}\!\left(x_{\beta^1}^{2,*}\right)=0,
\qquad
\Delta_{\beta^2}^{2}\!\left(x_{\beta^2}^{1,*}\right)=0.
\]
Define the associated grid points by
\[
\underline x^1:=\frac{\lfloor N^1 x_{\beta^2}^{1,*}\rfloor}{N^1},
\qquad
\bar x^1:=\frac{\lceil N^1 x_{\beta^2}^{1,*}\rceil}{N^1},
\]
and
\[
\underline x^2:=\frac{\lfloor N^2 x_{\beta^1}^{2,*}\rfloor}{N^2},
\qquad
\bar x^2:=\frac{\lceil N^2 x_{\beta^1}^{2,*}\rceil}{N^2}.
\]

For \(p=1,2\), let \(\underline x^{p,-}\) and \(\bar x^{p,+}\) denote the grid points immediately below and above \(x_{\beta^{-p}}^{p,*}\), respectively. Thus, if \(N^p x_{\beta^{-p}}^{p,*}\notin\mathbb Z\), then
\[
\underline x^{p,-}=\underline x^p,
\qquad
\bar x^{p,+}=\bar x^p.
\]
If \(N^p x_{\beta^{-p}}^{p,*}\in\mathbb Z\), then
\[
\underline x^{p,-}=\underline x^p-\frac1{N^p},
\qquad
\bar x^{p,+}=\bar x^p+\frac1{N^p}.
\]
The next proposition verifies Assumption~\ref{ass:2attractor}.

\begin{figure}
    \centering    
    \includegraphics[width=0.7\linewidth]{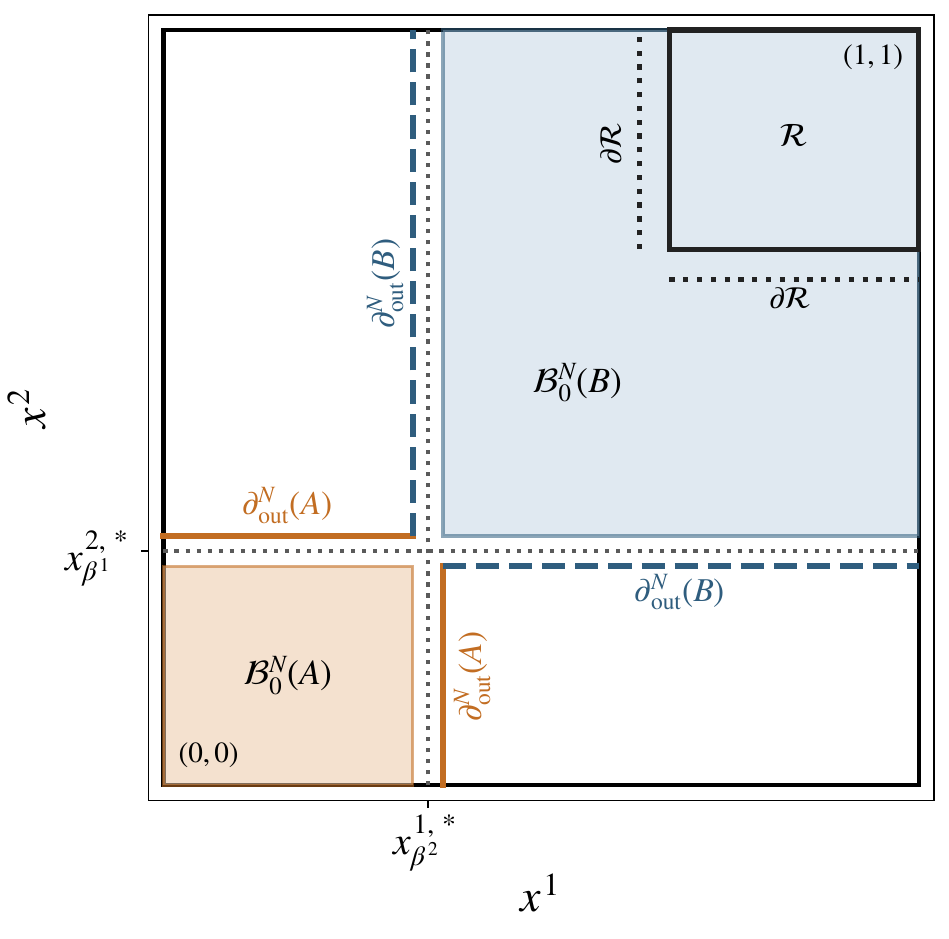}
    \caption{Zero-noise geometry for the two-population dynamics when the thresholds do not lie on the population grids, i.e., \(N^p x_{\beta^{-p}}^{p,*}\notin\mathbb Z\) for \(p=1,2\). The lower-left and upper-right rectangles are the stability sets \(\mathcal B_0^N(A)\) and \(\mathcal B_0^N(B)\), respectively, and the highlighted boundary segments indicate the exit boundaries \(\Exit^N(A)\) and \(\Exit^N(B)\).}
    \label{fig:2pop-geometry}
\end{figure}

\begin{proposition}\label{prop:2pop-ass2}
Under the two-population risk-sensitive dynamics induced by \eqref{eq:risk-rev}, with transition costs given by \eqref{eq:2pop-kappa}, the zero-noise dynamics have exactly two attracting states, namely \(A=(0,0)\) and \(B=(1,1)\). With \(N=(N^1,N^2)\), their stability sets are
\begin{align*}
\mathcal B_0^N(A)
=
\bigl([0,\underline x^{1,-}]\times[0,\underline x^{2,-}]\bigr)\cap X^N,\qquad 
\mathcal B_0^N(B)
=
\bigl([\bar x^{1,+},1]\times[\bar x^{2,+},1]\bigr)\cap X^N,
\end{align*}
and their exit boundaries are
\begin{align*}
\Exit^N(A)
&=
\Bigl[\bigl(\{\bar x^1\}\times[0,\underline x^{2,-}]\bigr)
\;\cup\;
\bigl([0,\underline x^{1,-}]\times\{\bar x^2\}\bigr)\Bigr]\cap X^N,\\
\Exit^N(B)
&=
\Bigl[\bigl(\{\underline x^1\}\times[\bar x^{2,+},1]\bigr)
\;\cup\;
\bigl([\bar x^{1,+},1]\times\{\underline x^2\}\bigr)\Bigr]\cap X^N.
\end{align*}
Consequently, Assumption~\ref{ass:2attractor} holds. The geometry is illustrated in Fig.~\ref{fig:2pop-geometry}.
\end{proposition}

\begin{proof}
By Lemma~\ref{lm:delta}, for each \(p\in\{1,2\}\),
\[
\Delta_{\beta^p}^p(z)<0 \quad \text{if } z<x_{\beta^p}^{-p,*},
\qquad
\Delta_{\beta^p}^p(z)>0 \quad \text{if } z>x_{\beta^p}^{-p,*},
\]
and \(\Delta_{\beta^p}^p(z)=0\) at \(z=x_{\beta^p}^{-p,*}\). Hence, by the two transition-cost identities in \eqref{eq:2pop-kappa}, if \(x^{-p}<x_{\beta^p}^{-p,*}\), then decreasing \(x^p\) by \(1/N^p\) has zero cost; if \(x^{-p}>x_{\beta^p}^{-p,*}\), then increasing \(x^p\) by \(1/N^p\) has zero cost; and if \(x^{-p}=x_{\beta^p}^{-p,*}\), then both moves have zero cost.

Define
\[
E_{00}:=\bigl([0,\underline x^{1,-}]\times[0,\underline x^{2,-}]\bigr)\cap X^N,
\qquad
E_{11}:=\bigl([\bar x^{1,+},1]\times[\bar x^{2,+},1]\bigr)\cap X^N.
\]
Every state in \(E_{00}\) reaches \(A=(0,0)\) along repeated downward zero-cost moves. Moreover, no state in \(E_{00}\) has a zero-cost path to \(X^N\setminus E_{00}\): any first exit from \(E_{00}\) must either increase \(x^1\) at a state with \(x^2\le \underline x^{2,-}<x_{\beta^1}^{2,*}\), or increase \(x^2\) at a state with \(x^1\le \underline x^{1,-}<x_{\beta^2}^{1,*}\), and each such move has strictly positive cost. Therefore,
\[
\mathcal B_0^N(A)=E_{00}.
\]
By the symmetric argument, every state in \(E_{11}\) reaches \(B=(1,1)\) along repeated upward zero-cost moves, and no state in \(E_{11}\) has a zero-cost path to \(X^N\setminus E_{11}\). Hence,
\[
\mathcal B_0^N(B)=E_{11}.
\]

We next show that all states outside these two stability sets are transient under the zero-noise dynamics. As depicted in Fig.~\ref{fig:2pop-geometry}, the sets \(E_{00}\) and \(E_{11}\) are the lower-left and upper-right rectangles, respectively. Consider any state \(x\notin E_{00}\). Then at least one coordinate is at or above its threshold. If \(x^1\ge \bar x^1\), then population \(2\) has a zero-cost upward move, so \(x^2\) can be increased until it reaches the upper region. Once \(x^2\ge \bar x^2\), population \(1\) has zero-cost upward moves, and the state can reach \(B=(1,1)\). The case \(x^2\ge \bar x^2\) is symmetric. Thus, every state outside \(E_{00}\) has a zero-cost path to \(B\).

Similarly, every state outside \(E_{11}\) has a zero-cost path to \(A\). Indeed, if \(x^1\le \underline x^1\), then population \(2\) has zero-cost downward moves, so \(x^2\) can be decreased until it reaches the lower region; once \(x^2\le \underline x^2\), population \(1\) has zero-cost downward moves, and the state can reach \(A=(0,0)\). The case \(x^2\le \underline x^2\) is symmetric. If a coordinate lies exactly on a threshold grid line, the corresponding population is indifferent, so both upward and downward moves in the other coordinate have zero cost; the same argument applies after one such zero-cost move.

Therefore, every state in \(X^N\setminus(E_{00}\cup E_{11})\) has zero-cost paths to both \(A\) and \(B\). Since \(A\) and \(B\) themselves have no outgoing zero-cost transitions, they are the only recurrent states of the zero-noise dynamics, i.e., the only attracting states.

The formulas for the exit boundaries follow directly from the descriptions of \(E_{00}\) and \(E_{11}\). From \(E_{00}\), the first one-step exits occur by increasing either \(x^1\) to \(\bar x^1\) while \(x^2\le \underline x^{2,-}\), or \(x^2\) to \(\bar x^2\) while \(x^1\le \underline x^{1,-}\). This gives the stated expression for \(\Exit^N(A)\). The expression for \(\Exit^N(B)\) follows analogously by considering the first downward exit from \(E_{11}\).

Finally, Assumption~\ref{ass:2attractor} \textup{(I)} holds because every state outside \(\{A,B\}\) has a zero-cost path to at least one of \(A\) or \(B\). Assumption~\ref{ass:2attractor} \textup{(II)} holds because every state outside \(\mathcal B_0^N(A)=E_{00}\) has a zero-cost path to \(B\), and every state outside \(\mathcal B_0^N(B)=E_{11}\) has a zero-cost path to \(A\).
\end{proof}

We define the escape costs from the two stability sets as
\begin{align}
\label{eq:escape-costs-2pop} 
\gamma_{00}^{N}:=
C^{N}\!\bigl((1,1),\Exit^{N}(B)\bigr),\qquad\qquad 
\gamma_{11}^{N}:=
C^{N}\!\bigl((0,0),\Exit^{N}(A)\bigr).
\end{align}
We also define the selection coefficient
\[
S^{N}:=\gamma_{11}^{N}-\gamma_{00}^{N}.
\]
By Theorem~\ref{thm:sse}, the state \((1,1)\) is uniquely stochastically stable if \(S^{N}<0\), whereas \((0,0)\) is uniquely stochastically stable if \(S^{N}>0\). 

\begin{theorem}[Stochastic stability under \texttt{2Pop-RS-BRM}]
\label{thm:2pop-rsbrm}
Consider an asymmetric two-population coordination game under \texttt{2Pop-RS-BRM}. For fixed \(N=(N^1,N^2)\), the escape costs defined in \eqref{eq:escape-costs-2pop} satisfy
\begin{align*}
\gamma_{00}^{N}
=
\min\!\bigl\{N^1(1-\underline x^1),\,N^2(1-\underline x^2)\bigr\},\qquad\qquad 
\gamma_{11}^{N}
=
\min\!\bigl\{N^1\bar x^1,\,N^2\bar x^2\bigr\}.
\end{align*}
Moreover, if \((1,1)\) is super-dominant, then for every fixed pair \((\beta^1,\beta^2)\), there exists some \(N_0(\beta^1,\beta^2)\in\mathbb N\) such that \((1,1)\) is uniquely stochastically stable for all \(N^1,N^2\ge N_0(\beta^1,\beta^2)\). The analogous statement holds for \((0,0)\).
\end{theorem} 
\begin{proof}
By the definition of \texttt{2Pop-RS-BRM}, the cost function in
Assumption~\ref{ass:cost} is \(c(d)=\mathbf 1_{\{d>0\}}\). Let \(B=(1,1)\). By Proposition~\ref{prop:2pop-ass2}, the exit boundary of
\(\mathcal B_0^{N^1,N^2}(B)\) is
\[
\Exit^{N^1,N^2}(B)
=
\Bigl[\bigl(\{\underline x^1\}\times[\bar x^{2,+},1]\bigr)
\;\cup\;
\bigl([\bar x^{1,+},1]\times\{\underline x^2\}\bigr)\Bigr]\cap X^{N^1,N^2}.
\]
Consider any path from \(B\) to \(\Exit^{N^1,N^2}(B)\), and truncate it at its
first hitting time of \(\Exit^{N^1,N^2}(B)\). Since transition costs are
nonnegative, this truncation cannot increase the path cost. Before the terminal
step, the path lies in \(\mathcal B_0^{N^1,N^2}(B)\).

If the truncated path ends in
\(\{\underline x^1\}\times[\bar x^{2,+},1]\), then its first coordinate must
have decreased from \(1\) to \(\underline x^1\). Hence, the path contains at least
\(N^1(1-\underline x^1)\) downward moves in population \(1\). Each such move
starts from a state in \(\mathcal B_0^{N^1,N^2}(B)\), where
\(x^2\ge \bar x^{2,+}>x_{\beta^1}^{2,*}\). Therefore, action \(1\) is the unique
ERM-adjusted best response for population \(1\), and a downward move is a strict
mistake with cost \(1\). Thus, the cost of such a path is at least
\(N^1(1-\underline x^1)\).

Similarly, if the truncated path ends in
\([\bar x^{1,+},1]\times\{\underline x^2\}\), then it contains at least
\(N^2(1-\underline x^2)\) downward moves in population \(2\), each of cost \(1\).
Therefore, every path from \(B\) to \(\Exit^{N^1,N^2}(B)\) has cost at least
\begin{align*}
\min\!\bigl\{N^1(1-\underline x^1),\,N^2(1-\underline x^2)\bigr\}.
\end{align*}
This lower bound is attained by the monotone path that decreases only the
coordinate achieving the minimum, while keeping the other coordinate equal to
\(1\). Hence
\[
\gamma_{00}^{N^1,N^2}
=
\min\!\bigl\{N^1(1-\underline x^1),\,N^2(1-\underline x^2)\bigr\}.
\]

The same argument applies to \(A=(0,0)\). Any path from \(A\) to \(\Exit^{N^1,N^2}(A)\), truncated at its first hitting
time of the exit boundary, must either increase \(x^1\) from \(0\) to
\(\bar x^1\), or increase \(x^2\) from \(0\) to \(\bar x^2\). While the path
remains in \(\mathcal B_0^{N^1,N^2}(A)\), each such upward move is a strict
mistake under \texttt{2Pop-RS-BRM} and has cost \(1\). The corresponding
monotone coordinate paths attain these costs. Therefore
\begin{align*}
\gamma_{11}^{N^1,N^2}
=
\min\!\bigl\{N^1\bar x^1,\,N^2\bar x^2\bigr\}.
\end{align*}

Finally, suppose that \((1,1)\) is super-dominant for both populations. Then, 
Lemma~\ref{lem:xstar-beta-monotonicity} implies
\[
x_{\beta^2}^{1,*}<\frac12,
\qquad
x_{\beta^1}^{2,*}<\frac12 .
\]
Since \(\bar x^p\to x_{\beta^{-p}}^{p,*}\) and
\(\underline x^p\to x_{\beta^{-p}}^{p,*}\) as \(N^p\to\infty\), there exists
\(N_0(\beta^1,\beta^2)\) such that, for all
\(N^1,N^2\ge N_0(\beta^1,\beta^2)\),
\[
\bar x^p < \frac12 < 1-\underline x^p,
\qquad p\in\{1,2\}.
\]
Consequently,
\[
N^p\bar x^p < N^p(1-\underline x^p),
\qquad p\in\{1,2\}.
\]
Taking minima over \(p\) gives
\[
\gamma_{11}^{N^1,N^2}
<
\gamma_{00}^{N^1,N^2}.
\]
By Theorem~\ref{thm:sse}, this implies that \((1,1)\) is uniquely
stochastically stable. The proof for \((0,0)\) is symmetric.
\end{proof}

Theorem~\ref{thm:2pop-rsbrm} shows that, under \texttt{2Pop-RS-BRM}, equilibrium
selection is determined by the relative magnitudes of the two escape costs
\(\gamma_{00}^{N^1,N^2}\) and \(\gamma_{11}^{N^1,N^2}\). These escape costs depend
on the ERM thresholds \(x_{\beta^2}^{1,*}\) and \(x_{\beta^1}^{2,*}\), and hence
on the risk-sensitivity parameters \((\beta^1,\beta^2)\). Therefore, changing
agents' risk attitudes can change the relative difficulty of escaping the two
coordinated equilibria, and may consequently alter the stochastically stable
equilibrium. We illustrate this risk-attitude tuning effect on equilibrium selection in the numerical
example in Section~\ref{sec:num_illus_theorem_52}.

\begin{theorem}[Stochastic stability under \texttt{2Pop-RS-logit}]
\label{thm:2pop_rslogit_superdominant}
Consider an asymmetric two-population \(2\times 2\) coordination game under
\texttt{2Pop-RS-logit}. Suppose that \((1,1)\) is super-dominant. Then, for every
fixed pair of risk parameters \((\beta^1,\beta^2)\in\mathbb R^2\), there exists
\(N_0(\beta^1,\beta^2)\) such that \((1,1)\) is uniquely stochastically stable
for all \(N^1,N^2\ge N_0(\beta^1,\beta^2)\). The analogous statement holds for
\((0,0)\).
\end{theorem}

\begin{proof}
We prove the claim for \((1,1)\); the proof for \((0,0)\) is symmetric. Fix
\((\beta^1,\beta^2)\in\mathbb R^2\), and write
\(A:=(0,0)\), \(B:=(1,1)\), and \(X:=X^{N^1,N^2}\) for notational convenience. 

Since \((1,1)\) is super-dominant for each population,
Lemma~\ref{lem:xstar-beta-monotonicity} gives
\[
x_{\beta^2}^{1,*}<\frac12,
\qquad
x_{\beta^1}^{2,*}<\frac12 .
\]
Therefore, for all sufficiently large \(N^1,N^2\), we have
\[
\bar x^{p,+}<1-\underline x^{p,-},
\qquad p\in\{1,2\}.
\]
Fix such \(N^1,N^2\). 
Recall from Proposition~\ref{prop:2pop-ass2} that
\[
\mathcal B_0^{N^1,N^2}(B)
=
\bigl([\bar x^{1,+},1]\times[\bar x^{2,+},1]\bigr)\cap X .
\]
Inside this stability set, define the rectangle
\[
\mathcal R
:=
\bigl([1-\underline x^{1,-},1]\times[1-\underline x^{2,-},1]\bigr)\cap X .
\]
By the preceding inequalities, \(\mathcal R\subseteq
\mathcal B_0^{N^1,N^2}(B)\). Define the south and west exit layer of
\(\mathcal R\) by
\[
\partial \mathcal R
:=
\Bigl[\bigl(\{1-\bar x^1\}\times[1-\underline x^{2,-},1]\bigr)
\cup
\bigl([1-\underline x^{1,-},1]\times\{1-\bar x^2\}\bigr)\Bigr]\cap X.
\]
The rectangle \(\mathcal R\) and its exit layer \(\partial\mathcal R\) are shown in Fig.~\ref{fig:2pop-geometry}.

Every path from \(B\in\mathcal R\) to \(\Exit^{N^1,N^2}(B)\) must leave
\(\mathcal R\), because \(\Exit^{N^1,N^2}(B)\cap\mathcal R=\emptyset\).
If such a path is truncated at its first exit from \(\mathcal R\), then the
terminal state of the truncated path lies in \(\partial\mathcal R\). Since all
one-step costs are nonnegative, truncation cannot increase the cost. Hence, 
\[
\gamma_{00}^{N^1,N^2}
=
C^{N^1,N^2}\bigl(B,\Exit^{N^1,N^2}(B)\bigr)
\ge
C^{N^1,N^2}\bigl(B,\partial\mathcal R\bigr).
\]

Now define the reflection map \(T:X\to X\) by
\[
T(x^1,x^2):=(1-x^1,1-x^2).
\]
Then, we have 
\[
T(\mathcal R)
=
\bigl([0,\underline x^{1,-}]\times[0,\underline x^{2,-}]\bigr)\cap X
=
\mathcal B_0^{N^1,N^2}(A),
\]
and
\[
T(\partial\mathcal R)
=
\Bigl[\bigl(\{\bar x^1\}\times[0,\underline x^{2,-}]\bigr)
\cup
\bigl([0,\underline x^{1,-}]\times\{\bar x^2\}\bigr)\Bigr]\cap X
=
\Exit^{N^1,N^2}(A).
\]

Let \(\psi=(\psi_0,\ldots,\psi_K)\) be a minimum-cost path from \((1,1)\) to \(\partial\mathcal R\), stopped at its first visit to this layer. Hence, \(\psi_i\in\mathcal R\) for every \(i<K\). Define the reflected path \(\phi=(\phi_0,\ldots,\phi_K)\) by
\(
\phi_i:=T(\psi_i), i=0,\ldots,K.
\)
Then \(\phi\) is a feasible path from \(A\) to \(\Exit^{N^1,N^2}(A)\), and \(\phi_i\in\mathcal B_0^{N^1,N^2}(A)\) for every \(i<K\).
We compare the two path costs step by step. Suppose first that a step of \(\psi\) moves left from \((u,v)\in\mathcal R\). Then
\(
v\ge1-\underline x^{2,-}>x_{\beta^1}^{2,*}\) and \(
1-v\le\underline x^{2,-}<x_{\beta^1}^{2,*}.
\)
Hence, the original left step has cost \(\Delta_{\beta^1}^1(v)\), whereas its reflected right step has cost \(-\Delta_{\beta^1}^1(1-v)\). By Lemma~\ref{lem:Rpos},
\(
\Delta_{\beta^1}^1(v)>-\Delta_{\beta^1}^1(1-v).
\)
Similarly, a down step from \((u,v)\in\mathcal R\) has cost \(\Delta_{\beta^2}^2(u)\), whereas its reflected up step has cost \(-\Delta_{\beta^2}^2(1-u)\), and
\(
\Delta_{\beta^2}^2(u)>-\Delta_{\beta^2}^2(1-u).
\)

If a step of \(\psi\) instead moves right or up, it has zero cost because its starting state lies in \(\mathcal R\). Its reflected left or down step also has zero cost because its starting state lies in
\(T(\mathcal R)=\mathcal B_0^{N^1,N^2}(A)\). Thus, reflection strictly reduces the cost of every left or down step and preserves the zero cost of every right or up step. Since \(\psi\) must contain at least one left or down step to reach \(\partial\mathcal R\), we obtain
\[
C^{N^1,N^2}(\phi)
<
C^{N^1,N^2}(\psi)
=
C^{N^1,N^2}\bigl((1,1),\partial\mathcal R\bigr).
\]
Since \(\phi\) is a path from \(A\) to \(\Exit^{N^1,N^2}(A)\),
\[
\gamma_{11}^{N^1,N^2}
\le
C^{N^1,N^2}(\phi)
<
C^{N^1,N^2}\bigl((1,1),\partial\mathcal R\bigr)
\le
\gamma_{00}^{N^1,N^2}.
\]
Thus, by Theorem~\ref{thm:sse}, \((1,1)\) is uniquely stochastically stable.
\end{proof}

Theorems~\ref{thm:2pop-rsbrm} and~\ref{thm:2pop_rslogit_superdominant} show
that the robustness result from the single-population setting extends to the
two-population setting. Under both revision protocols, a super-dominant equilibrium remains uniquely stochastically stable for all fixed risk-attitude
profiles \((\beta^1,\beta^2)\), provided the two population sizes are
sufficiently large. Thus, even when the two populations have heterogeneous risk
preferences, long-run equilibrium selection is pinned down by the payoff
structure rather than by the particular risk attitudes entering the learning
rule. The robustness of super-dominant equilibria is natural and useful because
super-dominance aligns payoff dominance with downside protection. Such an
equilibrium is attractive both to risk-seeking agents, who emphasize high-payoff
outcomes, and to risk-averse agents, who emphasize worst-case guarantees.
Therefore, heterogeneity in the risk parameters does not
change the direction of long-run selection: the payoff structure itself keeps the dynamics anchored at the super-dominant equilibrium.

\subsection{Numerical Illustration}\label{sec:num_illus_theorem_52}

\begin{figure}
    \centering
    \includegraphics[width=0.75\linewidth]{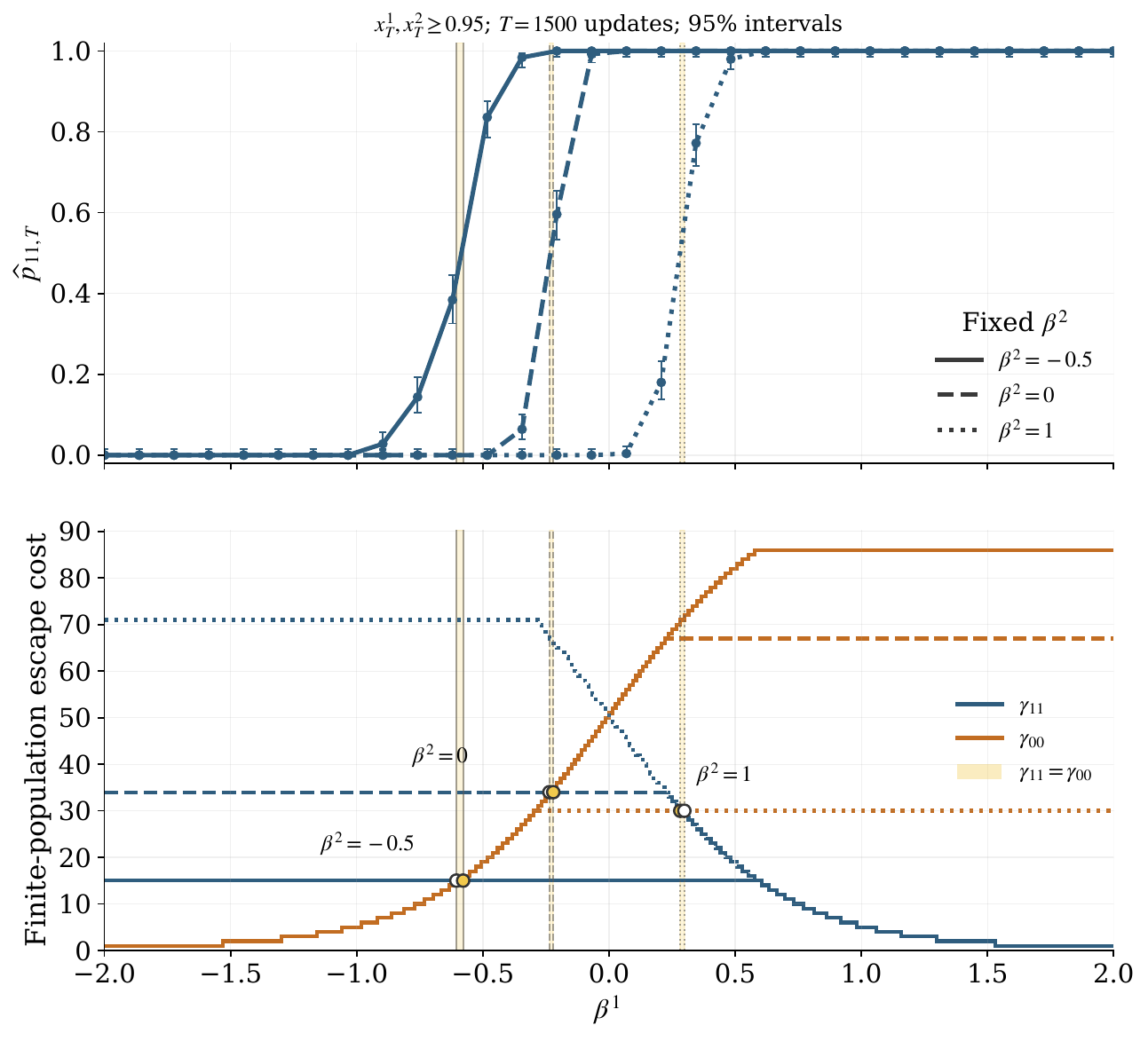}
    \vspace{-10pt}
    \caption{\texttt{2Pop-RS-BRM} as \(\beta^1\) varies, for three fixed values of \(\beta^2\). Top: finite-time coordination frequencies with pointwise \(95\%\) Wilson intervals. Bottom: finite-population escape costs; shading marks equal-cost regions.}
    \label{fig:2pop-rsbrm-numerical}
\end{figure}

We illustrate the dependence of finite-time coordination frequencies and
finite-population escape costs on risk attitudes under \texttt{2Pop-RS-BRM}.
Consider the asymmetric coordination game
\[
M^1=\begin{pmatrix}7 & 1\\ 4 & 4\end{pmatrix},
\qquad
M^2=\begin{pmatrix}5 & 4\\ 1 & 6\end{pmatrix}.
\]
This is a Battle-of-the-Sexes-type game \cite{osborne1994course}: population~1
receives a higher coordinated payoff at \((1,1)\), whereas population~2 receives
a higher coordinated payoff at \((0,0)\). Thus, the two populations have
conflicting payoff preferences over the two coordinated equilibria.

We fix \(N^1=N^2=100\), \(\eta=0.1\), and the initial state
\(x_0^1=x_0^2=0.5\). For each \(\beta^2\in\{-0.5,0,1\}\), we use
\(\beta_j^1=-2+4j/29\), \(j=0,\ldots,29\), and run \(250\) independent
Monte Carlo simulations per parameter pair. Each run lasts exactly
\(T=1500\) individual revision opportunities. We classify a terminal state
as near \((1,1)\) if both coordinates are at least \(0.95\),
as near \((0,0)\) if both are at most \(0.05\), and otherwise as ``other.''
The top panel of Fig.~\ref{fig:2pop-rsbrm-numerical} plots
\[
\widehat p_{11,T}
=\frac{1}{250}\sum_{r=1}^{250}
\mathbf 1_{\{x_{T,r}^1\geq0.95,\;x_{T,r}^2\geq0.95\}}.
\]
Error bars are pointwise \(95\%\) Wilson score intervals for this
terminal-event probability; near a
frequency of \(0.5\), their half-width is approximately \(0.062\).

The bottom panel plots the integer escape costs
\(\gamma_{11}^{N^1,N^2}\) and \(\gamma_{00}^{N^1,N^2}\) from
Theorem~\ref{thm:2pop-rsbrm}. These costs describe stochastic stability in the
vanishing-noise limit at fixed population sizes: \((1,1)\) is uniquely stochastically stable
when \(\gamma_{11}^{N^1,N^2}<\gamma_{00}^{N^1,N^2}\), and \((0,0)\) is uniquely
stochastically stable under the reverse strict inequality.

The terminal frequencies show how risk attitudes affect coordination from the
specified initial state over this fixed horizon. For each fixed value of
\(\beta^2\), the observed transition occurs near the parameter interval where
the two escape costs tie. The direction of this transition also has a clear
payoff interpretation for population~1. As \(\beta^1\) increases, population~1
becomes more risk-seeking, and more runs end near \((1,1)\), the coordinated
equilibrium with the higher payoff for population~1. Conversely, as
\(\beta^1\) decreases, population~1 becomes more risk-averse, and more runs end
near \((0,0)\), which corresponds to its maximin choice. Thus, changing only
population~1's risk attitude can change the relative difficulty of escaping
the two coordinated equilibria and shift finite-time coordination from its
maximin equilibrium toward the equilibrium with the higher payoff. Moreover,
the location of this transition depends on population~2's risk attitude,
showing that finite-time coordination is shaped jointly by the risk attitudes
of both populations.

\section{Extension: $k$-Action Setting}\label{sec:k_action}

We now extend the single-population analysis from two actions to symmetric games
with an arbitrary finite number of actions. 

\subsection{Model and Setup}

Let \(S=\{1,\ldots,k\}\), with
\(k\ge 2\), denote the common action set, and let
\(A=(a_{ij})_{i,j\in S}\) be the payoff matrix, where \(a_{ij}\) is the payoff
to an agent playing action \(i\) against an opponent playing action \(j\).

For a population of size \(N\), the finite state space is
\[
\Delta_k^N :=
\left\{
x\in \frac1N\mathbb Z_+^k:
\sum_{i=1}^k x_i=1
\right\},
\]
where \(x_i\) denotes the fraction of agents currently playing action \(i\).

In a symmetric \(k\)-action game, a diagonal action profile \((i,i)\) is a
strict Nash equilibrium if action \(i\) is the unique best response to itself,
that is,
\[
    a_{ii}>a_{ji},
    \qquad \forall j\in S\setminus\{i\}.
\]
The corresponding population state to this equilibrium is the vertex
\(e_i\in\Delta_k^N\), where all agents play action \(i\).

\begin{definition}[\(k\)-action coordination games]
\label{def:k-action-coordination}
A symmetric \(k\)-action game is a \emph{\(k\)-action coordination game} if
every diagonal action profile \((i,i)\), \(i\in S\), is a strict Nash
equilibrium; equivalently,
\[
    a_{ii}>a_{ji},
    \qquad \forall i\in S,\ \forall j\in S\setminus\{i\}.
\]
\end{definition}

Before proceeding to the stochastic-stability analysis, we introduce two refinement notions for diagonal Nash equilibria.

\begin{definition}[Payoff dominance]
\label{def:payoff-dominance}
In a symmetric \(k\)-action coordination game, a strict coordinated equilibrium \((i,i)\) is
\emph{payoff-dominant} if its coordinated payoff is strictly larger than the
coordinated payoff of every other strict coordinated equilibrium:
\[
    a_{ii}>a_{ss},
    \qquad \forall s\in S\setminus\{i\}.
\]
\end{definition}

\begin{definition}[Strong payoff dominance]
\label{def:strong-payoff-dominance}
In a symmetric \(k\)-action game, a diagonal profile \((i,i)\) is
\emph{strongly payoff-dominant} if
\[
    a_{ii}>
    \max_{s\in S\setminus\{i\}}\max_{r\in S} a_{sr}.
\]
In particular, any strongly payoff-dominant diagonal profile is a strict Nash
equilibrium, because the above condition implies
\[
    a_{ii}>a_{ji},
    \qquad \forall j\in S\setminus\{i\}.
\]
\end{definition}

Strong payoff dominance requires the diagonal payoff \(a_{ii}\) to exceed not
only the diagonal payoffs \(a_{ss}\) of all other actions \(s\neq i\), but also
all payoffs obtainable by any action \(s\neq i\) against any opponent action
\(r\in S\). In a \(k\)-action coordination game, this condition is \emph{equivalent to} 
payoff dominance. Indeed, suppose \((i,i)\) is payoff-dominant, so that
\(a_{ii}>a_{rr}\) for every \(r\neq i\). Since the game is a coordination game,
each diagonal profile \((r,r)\) is a strict Nash equilibrium. Hence, we have
\[
a_{rr}>a_{sr},
\qquad \forall s\neq r .
\]
Therefore, for every \(s\neq i\) and every \(r\in S\), we have
\(a_{ii}>a_{sr}\): if \(r=i\), this follows from the strictness of \((i,i)\);
if \(r\neq i\), it follows from \(a_{ii}>a_{rr}> a_{sr}\) when $s \neq r$ and \(a_{ii}>a_{rr} = a_{sr}\) when $s = r$. Thus, payoff dominance
implies strong payoff dominance. The converse is immediate by taking \(r=s\).
Hence, in \(k\)-action coordination games, payoff dominance and strong payoff
dominance coincide. The strong formulation is nevertheless useful because it is
defined directly in \emph{general} $k$-action symmetric games, and can be applied beyond the
coordination-game structure.

We also introduce two maximin-type refinements that will be used later.

\begin{definition}[Strict maximin]
\label{def:strict-maximin}
In a symmetric \(k\)-action game, action \(i\in S\) is \emph{strictly maximin}
if it strictly maximizes the worst-case payoff:
\[
    \min_{r\in S} a_{ir}
    >
    \min_{r\in S} a_{sr},
    \qquad \forall s\in S\setminus\{i\}.
\]
If, in addition, \((i,i)\) is a strict Nash equilibrium, then we call the
diagonal equilibrium \((i,i)\) strictly maximin.
\end{definition}

\begin{definition}[Strong maximin]
\label{def:strong-maximin}
In a symmetric \(k\)-action game, action \(i\in S\) is \emph{strongly maximin}
if
\[
    \min_{r\in S} a_{ir}
    >
    \max_{s\in S\setminus\{i\}} a_{si}.
\]
In this case, the diagonal profile \((i,i)\) is a strict Nash equilibrium,
because
\[
    a_{ii}\ge \min_{r\in S}a_{ir}
    >
    \max_{s\in S\setminus\{i\}}a_{si}
    \ge a_{ji},
    \qquad \forall j\in S\setminus\{i\}.
\]
\end{definition}

Strong maximin is a stronger condition than strict maximin. To see
this, suppose action \(i\) is strongly maximin. Then, for every \(s\neq i\),
\[
    \min_{r\in S} a_{sr}
    \le a_{si}
    \le \max_{s'\in S\setminus\{i\}} a_{s'i}
    <
    \min_{r\in S} a_{ir}.
\]
Hence, action \(i\) strictly maximizes the worst-case payoff, and is thus 
strictly maximin. The stronger formulation is useful as it compares the
worst-case payoff of action \(i\) with the best payoff that any alternative
action can obtain when the opponent plays \(i\).

\subsection{Risk-Sensitive Evolutionary Dynamics}

We now introduce the risk-sensitive evolutionary dynamics for the
single-population \(k\)-action setting. To keep the analysis focused, we consider
the best-response-with-mutations protocol, and refer to the resulting dynamics
as \texttt{RS-BRM}, consistent with the two-action case.

Given a state \(x\in\Delta_k^N\), the ERM-adjusted payoff of action \(i\in S\)
is
\[
F_{\beta,i}(x)
=
\frac1\beta
\log\!\left(\sum_{\ell=1}^k x_\ell e^{\beta a_{i\ell}}\right),
\qquad \beta\neq 0,
\]
with the continuous risk-neutral extension
\[
F_{0,i}(x)=\sum_{\ell=1}^k x_\ell a_{i\ell}.
\]
We write
\begin{align*}
F_\beta(x):=\bigl(F_{\beta,1}(x),\ldots,F_{\beta,k}(x)\bigr)
\end{align*}
for the vector of ERM-adjusted payoffs at state \(x\). The corresponding
ERM-adjusted best-response set is
\begin{align*}
    \mathrm{BR}_\beta(x)
    :=
    \arg\max_{j\in S} F_{\beta,j}(x).
\end{align*}

At each time step, one agent is selected uniformly at random and receives a
revision opportunity. Under \texttt{RS-BRM}, the revising agent chooses an
ERM-adjusted best response with probability \(1-e^{-1/\eta}\), and chooses a
non-best-response action with probability \(e^{-1/\eta}\), where
\(\eta>0\) is the noise level. More precisely, if
\(\mathrm{BR}_\beta(x)\neq S\), then the probability of choosing action
\(j\in S\) is
\[
q_j^\eta(x)
=
\begin{cases}
\dfrac{1-e^{-1/\eta}}{|\mathrm{BR}_\beta(x)|},
& j\in \mathrm{BR}_\beta(x), \\[1.2ex]
\dfrac{e^{-1/\eta}}{k-|\mathrm{BR}_\beta(x)|},
& j\notin \mathrm{BR}_\beta(x).
\end{cases}
\]
If all actions are tied, i.e., \(\mathrm{BR}_\beta(x)=S\), we set
\[
q_j^\eta(x)=\frac1k,
\qquad j\in S.
\]
Thus, in the vanishing-noise limit, only ERM-adjusted best responses are chosen
with non-exponentially small probability.

Since only one agent revises at a time, the feasible nontrivial one-step
transitions are
\[
    x \longrightarrow x+\frac1N(e_j-e_i),
    \qquad i,j\in S,\quad i\neq j,\quad x_i>0,
\]
where an agent currently playing action \(i\) switches to action \(j\). The corresponding one-step transition probability is
\[
P^{N,\eta}\!\left(x,x+\frac1N(e_j-e_i)\right)
=
x_i q_j^\eta(x).
\]
Here, \(x_i\) is the probability that the selected revising agent currently
plays action \(i\), and \(q_j^\eta(x)\) is the probability that this agent
chooses action \(j\) after revision.

Because \(x_i>0\) is independent of \(\eta\), it does not contribute to the
small-noise exponential cost. Therefore, the cost of a feasible nontrivial
transition is
\begin{align}\label{eq:kappa-k-brm}
    \kappa^N\!\left(x,x+\frac1N(e_j-e_i)\right)
=
\begin{cases}
0, & j\in \mathrm{BR}_\beta(x),\\
1, & j\notin \mathrm{BR}_\beta(x).
\end{cases}
\end{align}
Thus, under \texttt{RS-BRM}, transitions to ERM-adjusted best-response actions
have zero cost, whereas transitions to actions outside the ERM-adjusted
best-response set have unit cost. The associated zero-noise dynamics are 
defined by the directed graph on \(\Delta_k^N\) whose edges are precisely the
zero-cost transitions in \eqref{eq:kappa-k-brm}; equivalently, these are the
one-agent revisions to ERM-adjusted best responses.

\subsection{Analyses}

We now analyze long-run equilibrium selection under \texttt{RS-BRM} in the
single-population \(k\)-action setting. In this section, the finite state space
in Definition~\ref{def:zero-noise-stability-set} is \(X^N=\Delta_k^N\), and the zero-noise dynamics are the directed graph generated by zero-cost
\texttt{RS-BRM} transitions, equivalently, by one-agent revisions to
ERM-adjusted best responses. For any attracting state
\(A\in\Delta_k^N\), we use \(\mathcal B_0^N(A)\) and \(\Exit^N(A)\) as defined
in Definition~\ref{def:zero-noise-stability-set}.

\begin{definition}[Radius and coradius {\cite{ellison2000basins}}]
\label{def:radius-coradius}
Let \(A\in\Delta_k^N\) be an attracting state of the zero-noise
\texttt{RS-BRM} dynamics. The \emph{radius} of \(A\) is
\[
R^N(A):=
C^N\!\left(A,\Exit^N(A)\right),
\]
and the \emph{coradius} of \(A\) is
\[
CR^N(A):=
\max_{x\in \Delta_k^N\setminus \mathcal B_0^N(A)}
C^N\!\left(x,\mathcal B_0^N(A)\right).
\]
\end{definition}

We use the following radius--coradius criterion to identify stochastically
stable states for the risk-sensitive \texttt{RS-BRM} dynamics.

\begin{theorem}[Radius--coradius criterion for \texttt{RS-BRM}]
\label{thm:radius-coradius-brm}
Fix \(N\), and consider the Markov chain \(X^{N,\eta}\) induced by
\texttt{RS-BRM} on the finite state space \(\Delta_k^N\). Let
\(A\in\Delta_k^N\) be an attracting state of the zero-noise \texttt{RS-BRM}
dynamics. If
\[
    R^N(A)>CR^N(A),
\]
then \(A\) is the unique stochastically stable state.
\end{theorem}

\begin{proof}
For every \(\eta>0\), the \texttt{RS-BRM} learning dynamics define an
irreducible finite-state Markov chain. As \(\eta\to0\), this family of Markov
chains is a regular perturbation of the zero-noise \texttt{RS-BRM} dynamics.
Indeed, for a feasible transition
\[
x\to x+\frac1N(e_j-e_i),
\qquad i\neq j,\quad x_i>0,
\]
the transition probability is \(x_iq_j^\eta(x)\). Since \(x_i>0\) is independent
of \(\eta\), it does not affect the exponential rate. Hence, the resistance of
the transition is \(0\) if \(j\in\mathrm{BR}_\beta(x)\), and is \(1\) if
\(j\notin\mathrm{BR}_\beta(x)\). In particular, moves to ERM-adjusted best
responses, including moves along payoff ties, have zero resistance, while
moves to non-best-response actions have unit resistance.

Thus, the unperturbed dynamics in Ellison's framework coincide with the
zero-noise \texttt{RS-BRM} dynamics defined above. Taking \(\Omega=\{A\}\),
Ellison's basin, radius, and coradius coincide respectively with
\(\mathcal B_0^N(A)\), \(R^N(A)\), and \(CR^N(A)\). The claim then follows from
the radius--coradius theorem of Ellison~\cite[Theorem~1]{ellison2000basins}.
\end{proof}

\begin{theorem}
\label{thm:payoff-dominance-high-beta}
Consider the single-population risk-sensitive evolutionary dynamics under
\texttt{RS-BRM} in a finite symmetric \(k\)-action game with action set
\(S=\{1,\ldots,k\}\). Let \(e_i\in\Delta_k^N\) denote the state in
which all agents play action \(i\in S\).

Suppose first that the diagonal profile \((i,i)\) is strongly payoff-dominant,
and define
\[
\delta_i
:=
a_{ii}
-
\max_{s\in S\setminus\{i\}}\max_{r\in S}a_{sr}
>0 .
\]
Let $\bar\beta_i:=\frac{\log 2}{\delta_i}$. 
 Then, for every \(\beta>\bar\beta_i\), there exists
\(N_0(\beta)\) such that, for every \(N\ge N_0(\beta)\), the state \(e_i\) is
uniquely stochastically stable under \texttt{RS-BRM}.

Similarly, suppose that the diagonal profile \((i,i)\) is strongly maximin, and
define
\[
    \delta_i^{\mathrm{mm}}
    :=
    \min_{r\in S} a_{ir}
    -
    \max_{s\in S\setminus\{i\}} a_{si}
    >0 .
\]
Let $\underline\beta_i
    :=
    -\frac{\log 2}{\delta_i^{\mathrm{mm}}}.$ 
Then, for every \(\beta<\underline\beta_i\), there exists
\(N_0(\beta)\) such that, for every \(N\ge N_0(\beta)\), the state \(e_i\) is
uniquely stochastically stable under \texttt{RS-BRM}.
\end{theorem}

\begin{proof}
We prove the two claims using the same argument. Let
\[
\Delta_k:=\left\{x\in\mathbb R_+^k:\sum_{r=1}^k x_r=1\right\}
\]
denote the continuous simplex. We first show that, in either case, there exists
\(q<1/2\) such that action \(i\) is the unique ERM-adjusted best response at
every state \(x\in\Delta_k\) satisfying \(x_i\ge q\). The desired statement then
applies in particular to all finite-population states
\(x\in\Delta_k^N\subset\Delta_k\).

First, consider the strong payoff-dominance case and fix
\(\beta>\bar\beta_i\). Then, we have
\[
    e^{-\beta\delta_i}<\frac12 .
\]
Choose \(q\) such that
\[
    e^{-\beta\delta_i}<q<\frac12 .
\]
Equivalently,
\[
    a_{ii}+\frac1\beta\log q
    >
    a_{ii}-\delta_i .
\]
Now take any \(x\in\Delta_k\) with \(x_i\ge q\). Since \(\beta>0\),
\[
F_{\beta,i}(x)
=
\frac1\beta
\log\!\left(\sum_{r\in S}x_r e^{\beta a_{ir}}\right)
\ge
\frac1\beta\log\!\left(x_i e^{\beta a_{ii}}\right)
=
a_{ii}+\frac1\beta\log x_i
\ge
a_{ii}+\frac1\beta\log q .
\]
Therefore,
\[
    F_{\beta,i}(x)>a_{ii}-\delta_i .
\]
On the other hand, for every \(s\neq i\),
\[
\sum_{r\in S}x_r e^{\beta a_{sr}}
\le
\sum_{r\in S}x_r e^{\beta \max_{r'\in S}a_{sr'}}
=
e^{\beta \max_{r'\in S}a_{sr'}} .
\]
Thus,
\[
F_{\beta,s}(x)
\le
\max_{r\in S}a_{sr}
\le
a_{ii}-\delta_i .
\]
Hence,
\[
    F_{\beta,i}(x)>F_{\beta,s}(x),
    \qquad \forall s\neq i .
\]
So action \(i\) is the unique ERM-adjusted best response whenever \(x_i\ge q\).

Now consider the strong maximin case and fix
\(\beta<\underline\beta_i\). Then, we have  $e^{\beta\delta_i^{\mathrm{mm}}}<\frac12.$ 
Choose \(q\) such that
\[
    e^{\beta\delta_i^{\mathrm{mm}}}<q<\frac12 .
\]
Since \(\beta<0\), this implies that $\frac1\beta\log q < \delta_i^{\mathrm{mm}}.$  
Let
\[
    m_i:=\min_{r\in S}a_{ir}.
\]
For any \(x\in\Delta_k\), since \(a_{ir}\ge m_i\) for all \(r\in S\) and
\(\beta<0\), we have
\[
e^{\beta a_{ir}}\le e^{\beta m_i},
\qquad \forall r\in S.
\]
Therefore,
\[
\sum_{r\in S}x_r e^{\beta a_{ir}}
\le
e^{\beta m_i}.
\]
Taking logarithms and then dividing by the negative number \(\beta\) reverses the
inequality, yielding $F_{\beta,i}(x)\ge m_i.$ 

Now fix any \(s\neq i\). Since all terms in the sum are nonnegative,
\[
\sum_{r\in S}x_r e^{\beta a_{sr}}
\ge
x_i e^{\beta a_{si}} .
\]
Taking logarithms and dividing by \(\beta<0\), we get
\[
F_{\beta,s}(x)
\le
a_{si}+\frac1\beta\log x_i .
\]
Since \(x_i\ge q\) and \(\beta<0\),
\[
\frac1\beta\log x_i
\le
\frac1\beta\log q .
\]
Moreover, by the definition of \(\delta_i^{\mathrm{mm}}\),
\[
a_{si}
\le
\max_{s'\neq i}a_{s'i}
=
m_i-\delta_i^{\mathrm{mm}} .
\]
Combining these inequalities gives
\[
F_{\beta,s}(x)
\le
m_i-\delta_i^{\mathrm{mm}}
+
\frac1\beta\log q
<
m_i
\le
F_{\beta,i}(x).
\]
Therefore,
\[
    F_{\beta,i}(x)>F_{\beta,s}(x),
    \qquad \forall s\neq i .
\]
Thus, in both cases, there exists \(q<1/2\) such that action \(i\) is the unique
ERM-adjusted best response whenever \(x_i\ge q\).

Now define
\begin{align*}
    H_q^N:=\{x\in \Delta_k^N:x_i\ge q\}.
\end{align*}
Since action \(i\) is the unique ERM-adjusted best response at every state in
\(H_q^N\), any one-agent revision from an action \(s\neq i\) to action \(i\) has
zero cost. Repeating such revisions gives a zero-cost path from every state in
\(H_q^N\) to \(e_i\). Moreover, no zero-cost path can leave \(H_q^N\). Indeed,
leaving \(H_q^N\) requires a transition that decreases \(x_i\) from a state with
\(x_i\ge q\) to a state with \(x_i<q\). Such a transition switches an agent away
from action \(i\), even though action \(i\) is the unique ERM-adjusted best
response, and hence has cost \(1\) under \texttt{RS-BRM}. Therefore,
\begin{align*}
    H_q^N\subseteq \mathcal B_0^N(e_i). 
\end{align*}
In particular, \(e_i\) is an attracting state of the zero-noise
\texttt{RS-BRM} dynamics.

If \(\mathcal B_0^N(e_i)=\Delta_k^N\), then \(e_i\) is the unique recurrent class
of the zero-noise \texttt{RS-BRM} dynamics, and the conclusion follows directly.
Otherwise, since \(H_q^N\subseteq\mathcal B_0^N(e_i)\), any path from \(e_i\) to
\(\Exit^N(e_i)\) must first leave \(H_q^N\). Starting from \(e_i\), this requires
decreasing \(x_i\) from \(1\) to a value strictly below \(q\). Hence, the path
must contain at least
\begin{align*}
    N-\lceil qN\rceil+1
\end{align*}
unit-cost revisions away from action \(i\). Therefore,
\begin{align*}
    R^N(e_i)
    \ge
    N-\lceil qN\rceil+1 .
\end{align*}

Conversely, from any state outside \(\mathcal B_0^N(e_i)\), one can reach
\(H_q^N\subseteq\mathcal B_0^N(e_i)\) by revising agents to action \(i\) one at a
time. Starting from any state, at most \(\lceil qN\rceil\) such revisions are
needed to make \(x_i\ge q\). Each such revision has cost at most \(1\) under
\texttt{RS-BRM}. Hence,
\[
    CR^N(e_i)\le \lceil qN\rceil .
\]

Since \(q<1/2\), there exists \(N_0(\beta)\) such that, for every
\(N\ge N_0(\beta)\),
\[
    N-\lceil qN\rceil+1>\lceil qN\rceil .
\]
For such
\(N\), we have
\[
    R^N(e_i)>CR^N(e_i).
\]
By Theorem~\ref{thm:radius-coradius-brm}, \(e_i\) is the unique stochastically
stable state. This completes the proof.
\end{proof}

We next provide a numerical illustration of how risk attitudes can affect the 
long-run equilibrium selection in the \(k\)-action setting.

\subsection{Numerical Illustration}

We illustrate Theorem~\ref{thm:payoff-dominance-high-beta} using the following
symmetric four-action game:
\[
A=
\begin{pmatrix}
-10 & -10 & -10 & 0\\
-10 & 10  & -10 & 0\\
4   & 4   & 7   & 0\\
2   & 2   & 2   & 5
\end{pmatrix}.
\]
This example is chosen to highlight the two opposite risk-sensitive selection
effects in Theorem~\ref{thm:payoff-dominance-high-beta}. In particular, one
diagonal equilibrium is strongly payoff-dominant and is therefore selected under
sufficiently risk-seeking behavior, while another diagonal equilibrium is
strongly maximin and is selected under sufficiently risk-averse behavior.

This game is a symmetric four-action game, but it is not a four-action
coordination game in the sense of Definition~\ref{def:k-action-coordination},
because \((1,1)\) is not a strict Nash equilibrium. By contrast, \((2,2)\), \((3,3)\),
and \((4,4)\) are strict diagonal equilibria. As before, \(e_i\in\Delta_4^N\)
denotes the state in which all agents play action \(i\).

The diagonal equilibrium \((2,2)\) is strongly payoff-dominant. Its payoff is
\(a_{22}=10\), while the largest payoff attainable by any action other than
action \(2\) is
\[
\max_{s\neq 2}\max_{r\in S}a_{sr}
=
\max\{0,7,5\}
=
7 .
\]
Therefore,
\[
\delta_2
=
a_{22}
-
\max_{s\neq 2}\max_{r\in S}a_{sr}
=
10-7
=
3>0.
\]
The corresponding high-risk-seeking threshold is
\[
\bar\beta_2
=
\frac{\log 2}{\delta_2}
=
\frac{\log 2}{3}
\approx 0.231.
\]
Hence, for \(\beta>\bar\beta_2\), Theorem~\ref{thm:payoff-dominance-high-beta}
predicts that \(e_2\) is the unique stochastically stable state for all
sufficiently large \(N\).

The diagonal equilibrium \((4,4)\) is strongly maximin. Action \(4\) guarantees
the payoff
\begin{align*}
\min_{r\in S}a_{4r}
=
\min\{2,2,2,5\}
=
2,
\end{align*}
whereas the largest payoff obtained by any other action against an opponent
playing action \(4\) is
\[
\max_{s\neq 4}a_{s4}
=
\max\{0,0,0\}
=
0.
\]
Thus,
\[
\delta_4^{\mathrm{mm}}
=
\min_{r\in S}a_{4r}
-
\max_{s\neq 4}a_{s4}
=
2-0
=
2>0.
\]
The corresponding low-risk threshold is
\[
\underline\beta_4
=
-\frac{\log 2}{\delta_4^{\mathrm{mm}}}
=
-\frac{\log 2}{2}
\approx -0.347.
\]
Therefore, for \(\beta<\underline\beta_4\),
Theorem~\ref{thm:payoff-dominance-high-beta} predicts that \(e_4\) is the
unique stochastically stable state for all sufficiently large \(N\).

\begin{figure}[t]
    \centering
    \includegraphics[width=0.82\linewidth]{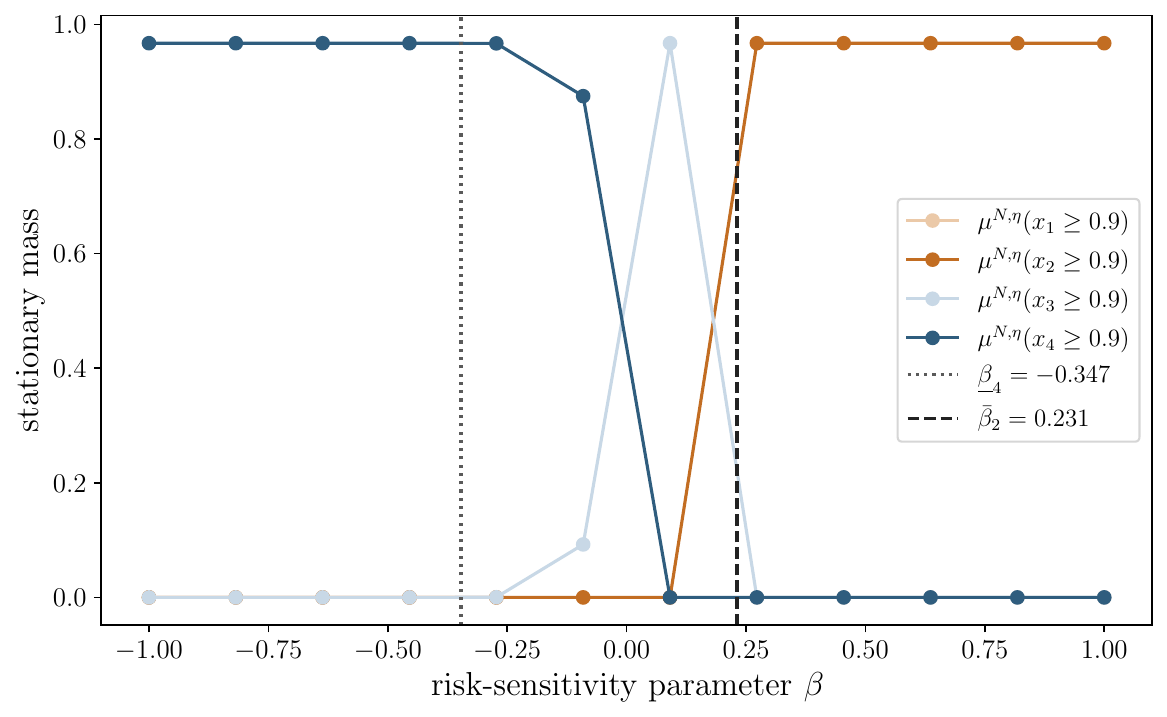}
    \caption{Stationary mass near each pure population state under \texttt{RS-BRM},
    with \(N=20\), \(\eta=0.3\), and threshold \(x_i\geq 0.9\).}
    \label{fig:near-stationary-mass}
\end{figure}

We compute the exact stationary distribution of the finite Markov chain induced
by \texttt{RS-BRM}. The experiment uses \(N=20\), \(\eta=0.3\), and a grid of
\(12\) evenly spaced values of \(\beta\in[-1,1]\). Since the stationary
distribution at a nonzero noise level may place substantial mass on states that
are close to, but not exactly equal to, an attracting state, we also report an
aggregate concentration measure. Specifically, Fig.~\ref{fig:near-stationary-mass}
plots $\mu^{N,\eta}\bigl(x_i\geq 0.9\bigr)$, 
the stationary probability that at least \(90\%\) of the population plays action
\(i\). This statistic gives a more robust finite-noise visualization of how the
stationary distribution concentrates around each attracting state.

This example illustrates how the risk-sensitivity parameter shapes long-run
selection. When \(\beta>0\), the ERM-adjusted payoff places greater emphasis on
high-payoff realizations. As \(\beta\) increases, the high diagonal payoff
\(a_{22}=10\) makes action \(2\) increasingly attractive, leading to selection
of the strongly payoff-dominant equilibrium \((2,2)\). When \(\beta<0\), the
ERM-adjusted payoff places greater emphasis on low-payoff realizations, so
selection is driven more by downside protection. In this example, action \(4\)
guarantees a payoff of at least \(2\) against every opponent action, while every
other action obtains at most \(0\) against an opponent playing action \(4\).
Thus, stronger risk aversion enables the evolutionary dynamics to favor the
more robust alternative against adverse opponent behavior, namely the strongly
maximin equilibrium \((4,4)\).

\section{Concluding Remarks} 

We studied risk-sensitive evolutionary learning dynamics in coordination games. By incorporating the entropic risk measure into noisy best-response protocols, we showed that agents' risk attitudes could change the stochastically stable outcome. In single-population symmetric \(2\times2\) coordination games, we showed that greater risk-seeking behavior favors the payoff-dominant equilibrium, whereas greater risk-averse behavior favors the maximin (and thus worst-case robust) equilibrium whenever these two equilibria differ. We also showed that this selection pattern holds under both protocols of  best response with mutations and logit choice.

In both the single-population and two-population settings, we also identified a robust regime in which super-dominant equilibria remain stochastically stable for all risk-sensitivity parameters. Thus, when payoff dominance and worst-case robustness are \emph{aligned}, long-run selection can be insensitive to agents' risk attitudes and to the specific noisy best-response protocol. Finally, we extended the single-population analysis to symmetric $k$-action games, 
and showed that sufficiently risk-seeking agents select strongly payoff-dominant equilibria, while sufficiently risk-averse agents select strongly maximin equilibria, whenever such equilibria exist. These results together showed that entropic risk sensitivity may serve as a systematic mechanism for steering equilibrium selection beyond the classical risk-neutral benchmark.

Several directions remain open. First, our analysis focused on homogeneous risk attitudes \emph{within} each population. Relaxing this assumption and allowing risk sensitivity to vary across agents or over time could provide a broader understanding of how risk preferences can shape the long-run equilibrium selection.
Second, we studied entropic risk sensitivity because it provided a natural and analytically tractable way to incorporate payoff uncertainty into evolutionary learning. Other risk-sensitive criteria may induce different effective payoff comparisons and, consequently, different stochastic-stability patterns, which are worth further investigation. 
Finally, the present analysis focused on coordination-game environments. Extending the framework to richer population structures and other  payoff structures  
may further illustrate how risk attitudes can shape the collective behaviors in multi-agent interactions. 

\section*{Acknowledgement}
The authors acknowledge the  support from the Army Research Office (ARO) grant W911NF-24-1-0085, the NSF CAREER Award 2443704, the AFOSR YIP Award FA9550-25-1-0258, an AI Safety
Research Award from Coefficient Giving,  a Cisco Faculty Research Award, and a JP Morgan Faculty Research Award. The authors also acknowledge the conversations with Tamer Ba\c{s}ar, Xiangyu Liu, and Eric Mazumdar, and the valuable feedback from the anonymous reviewers of the IEEE Conference on Decision and Control (CDC) 2026. 

\begin{appendices}

\section{Deferred Proofs}\label{secA1}

\subsection{Proof of Lemma~\ref{lm:delta}}\label{app:lem-delta}

\begin{proof}
Suppress the population superscript \(p\), and write \(\beta,\Delta,a,b,c,d\). Fix \(z\in(0,1)\). For \(\beta\neq0\),
\[
\Delta'(z)
=
\frac{1}{\beta}
\left(
\frac{e^{\beta a}-e^{\beta b}}{z e^{\beta a}+(1-z)e^{\beta b}}
-
\frac{e^{\beta c}-e^{\beta d}}{z e^{\beta c}+(1-z)e^{\beta d}}
\right).
\]
Let
\[
\alpha_1=e^{\beta(a-b)}, 
\qquad 
\alpha_2=e^{\beta(c-d)},
\qquad
\phi(\alpha):=\frac{\alpha-1}{z\alpha+1-z}.
\]
Then
\[
\Delta'(z)=\frac1\beta\bigl[\phi(\alpha_1)-\phi(\alpha_2)\bigr],
\qquad
\phi'(\alpha)=\frac{1}{(z\alpha+1-z)^2}>0.
\]
Hence, \(\phi\) is strictly increasing in \(\alpha\). Therefore,
\[
\operatorname{sgn}\Delta'(z)
=
\operatorname{sgn}(\beta)\operatorname{sgn}(\alpha_1-\alpha_2).
\]
If \(\beta>0\), then
\[
\operatorname{sgn}(\alpha_1-\alpha_2)
=
\operatorname{sgn}\bigl((a-b)-(c-d)\bigr),
\]
whereas if \(\beta<0\), this sign is reversed. This reversal is exactly offset by the negative factor \(1/\beta\). Thus, in both cases,
\[
\operatorname{sgn}\Delta'(z)
=
\operatorname{sgn}\bigl((a-b)-(c-d)\bigr).
\]
Since \(a>c\) and \(d>b\) in a coordination game, we have
\[
(a-b)-(c-d)=(a-c)+(d-b)>0.
\]
Hence, \(\Delta'(z)>0\) for all \(z\in(0,1)\), so \(\Delta\) is strictly increasing on \((0,1)\).

For \(\beta=0\), the expected-payoff expression gives
\[
\Delta_0(z)
=
z(a-c)+(1-z)(b-d),
\]
whose derivative is
\[
\Delta_0'(z)=(a-c)+(d-b)>0.
\]
Thus, strict monotonicity also holds at \(\beta=0\).

Moreover,
\[
\Delta(0)=b-d<0,
\qquad
\Delta(1)=a-c>0.
\]
By continuity, \(\Delta\) has a root in \((0,1)\), and strict monotonicity makes this root unique. For \(\beta\neq0\), solving \(\Delta(z)=0\) yields
\[
z e^{\beta a}+(1-z)e^{\beta b}
=
z e^{\beta c}+(1-z)e^{\beta d},
\]
and hence \eqref{eq:xstar-beta}. The expression at \(\beta=0\) follows by taking the continuous extension, or equivalently by solving \(\Delta_0(z)=0\). The stated sign characterization follows immediately from strict monotonicity.
\end{proof}

\subsection{Proof of Lemma~\ref{lem:xstar-beta-monotonicity}}\label{app:lem-xstar}

\begin{proof}
The argument is identical for every population, so we suppress the superscript \(p\) and write
\[
x_\beta^*
=
\frac{e^{\beta d}-e^{\beta b}}
{\bigl(e^{\beta a}-e^{\beta c}\bigr)+\bigl(e^{\beta d}-e^{\beta b}\bigr)}.
\]
For \(\beta\neq 0\), the denominator is nonzero because \(e^{\beta a}-e^{\beta c}\) and \(e^{\beta d}-e^{\beta b}\) have the same sign. Hence, \(\beta\mapsto x_\beta^*\) is continuous on \(\mathbb R\setminus\{0\}\). At \(\beta=0\), using \(e^{\beta u}=1+\beta u+o(\beta)\), we obtain
\begin{align*}
x_\beta^*
\to
\frac{d-b}{(a-c)+(d-b)}.
\end{align*}
Therefore, \(\beta\mapsto x_\beta^*\) is continuous on \(\mathbb R\).

For \(\beta\neq 0\), differentiating gives
\begin{equation}\label{eq:dxstar-beta-generalized}
\frac{d x^{*}_{\beta}}{d\beta}
=
x^{*}_{\beta}\bigl(1-x^{*}_{\beta}\bigr)
\bigl[w_2(\beta)-w_1(\beta)\bigr],
\end{equation}
where
\[
w_1(\beta):=
\frac{a e^{\beta a}-c e^{\beta c}}{e^{\beta a}-e^{\beta c}},
\qquad
w_2(\beta):=
\frac{d e^{\beta d}-b e^{\beta b}}{e^{\beta d}-e^{\beta b}}.
\]
Let
\begin{align*}
F(\beta):=w_2(\beta)-w_1(\beta),
\qquad
s:=a-c>0,
\qquad
t:=d-b>0.
\end{align*}
A direct differentiation shows that
\[
F'(\beta)=g(s)-g(t),
\qquad
g(z):=\frac{z^2 e^{\beta z}}{(e^{\beta z}-1)^2},
\quad z>0.
\]
For fixed \(\beta\neq0\), differentiating \(g\) with respect to \(z\) gives
\[
g'(z)
=
-\frac{z e^{\beta z}}{(e^{\beta z}-1)^3}\,h(\beta z),
\qquad
h(r):=(r-2)e^r+r+2.
\]
We now show that \(g'(z)<0\) for all \(z>0\). Since
\[
h(0)=0,
\qquad
h'(r)=(r-1)e^r+1,
\]
and
\[
h''(r)=r e^r,
\]
the function \(h'\) decreases on \((-\infty,0]\), increases on \([0,\infty)\), and attains its minimum \(h'(0)=0\). Hence, \(h'(r)\ge0\) for all \(r\), with strict inequality for \(r\neq0\). Thus, \(h\) is strictly increasing and, since \(h(0)=0\), \(h(r)\) has the same sign as \(r\). The term \((e^r-1)^3\) also has the same sign as \(r\), so \(g'(z)<0\) for all \(z>0\).

It follows that \(F\) is constant when \(s=t\), and otherwise is strictly monotone on \((-\infty,0)\) and on \((0,\infty)\), with the same monotonicity direction on both intervals. As \(\beta\to0\),
\[
w_1(\beta)=\frac1\beta+\frac{a+c}{2}+O(\beta),
\qquad
w_2(\beta)=\frac1\beta+\frac{b+d}{2}+O(\beta),
\]
so the poles cancel and \(F\) extends continuously to \(\beta=0\), with \(F(0)=\frac{b+d-a-c}{2}\). Therefore, \(F\) is monotone on all of \(\mathbb R\).

Using \(a>c\) and \(d>b\), we have
\[
\lim_{\beta\to\infty}w_1(\beta)=a,
\qquad
\lim_{\beta\to\infty}w_2(\beta)=d,
\]
and
\[
\lim_{\beta\to-\infty}w_1(\beta)=c,
\qquad
\lim_{\beta\to-\infty}w_2(\beta)=b.
\]
Hence,
\[
\lim_{\beta\to\infty}F(\beta)=d-a,
\qquad
\lim_{\beta\to-\infty}F(\beta)=b-c.
\]
Moreover,
\[
\lim_{\beta\to\infty}x_\beta^*
=
\begin{cases}
1,& d>a,\\
0,& d<a,
\end{cases}
\qquad
\lim_{\beta\to-\infty}x_\beta^*
=
\begin{cases}
1,& b<c,\\
0,& b>c.
\end{cases}
\]

Since \(x_\beta^*(1-x_\beta^*)>0\), \eqref{eq:dxstar-beta-generalized} implies
\[
\operatorname{sgn}\!\left(\frac{d x_\beta^*}{d\beta}\right)
=
\operatorname{sgn}\bigl(F(\beta)\bigr).
\]
If \(b>c\) and \(d>a\), then both endpoint limits of \(F\) are positive. Since \(F\) is monotone, \(F(\beta)>0\) for all \(\beta\), so \(x_\beta^*\) is strictly increasing. The stated limits give case \textup{(I)}. Similarly, if \(b<c\) and \(d<a\), then both endpoint limits of \(F\) are negative. Hence, \(F(\beta)<0\) for all \(\beta\), so \(x_\beta^*\) is strictly decreasing, and the stated limits give case \textup{(II)}.

It remains to prove cases \textup{(III)} and \textup{(IV)}. Since \(a>c\) and \(d>b\), the quantities
\(e^{\beta a}-e^{\beta c}\) and \(e^{\beta d}-e^{\beta b}\) have the same sign as \(\beta\). Hence, for \(\beta\neq0\), the expression in \eqref{eq:xstar-beta} can be written as
\[
x_\beta^*
=
\frac{B_\beta}{A_\beta+B_\beta},
\]
where
\[
A_\beta:=\bigl|e^{\beta a}-e^{\beta c}\bigr|,
\qquad
B_\beta:=\bigl|e^{\beta d}-e^{\beta b}\bigr|.
\]
Then \(x_\beta^*<\tfrac12\) if and only if \(B_\beta<A_\beta\), and \(x_\beta^*>\tfrac12\) if and only if \(B_\beta>A_\beta\). For any \(u<v\),
\[
\bigl|e^{\beta v}-e^{\beta u}\bigr|
=
|\beta|\int_u^v e^{\beta r}\,dr.
\]
If \(b\ge c\) and \(d\le a\), with at least one strict inequality, then \([b,d]\) is a proper subinterval of \([c,a]\). Therefore,
\[
B_\beta
=
|\beta|\int_b^d e^{\beta r}\,dr
<
|\beta|\int_c^a e^{\beta r}\,dr
=
A_\beta.
\]
Thus, \(x_\beta^*<\tfrac12\) for all \(\beta\neq0\). At \(\beta=0\), \(a-c>d-b\) gives \(x_0^*=\frac{d-b}{(a-c)+(d-b)}<\tfrac12\). This proves \textup{(III)}.
Similarly, if \(b\le c\) and \(d\ge a\), with at least one strict inequality, then \([c,a]\) is a proper subinterval of \([b,d]\). Hence, \(A_\beta<B_\beta\), so \(x_\beta^*>\tfrac12\) for all \(\beta\neq0\). At \(\beta=0\), \(d-b>a-c\) gives \(x_0^*>\tfrac12\). This proves \textup{(IV)}.

Finally, if \(a=d\) and \(b=c\), then \eqref{eq:xstar-beta} simplifies to \(x_\beta^*\equiv\tfrac12\).
\end{proof}

\subsection{Proof of Theorem~\ref{thm:sse}}
\label{app:thm-sse-proof}

\begin{proof}
Let \(P^{N,\eta}\) be the transition matrix of \(X^{N,\eta}\). By the Markov-chain tree representation for finite irreducible Markov chains, also used in Freidlin--Wentzell stochastic-stability theory \cite{freidlin1998random,kandori1993learning,young1993evolution}, the stationary distribution satisfies
\[
\mu^{N,\eta}(x)
=
\frac{\tau_x^{N,\eta}}{\sum_{z\in X^N}\tau_z^{N,\eta}},
\]
where
\begin{align*}
\tau_x^{N,\eta}
:=
\sum_{T\in\mathcal T_x}
\prod_{(u\to v)\in T}P^{N,\eta}(u,v),
\end{align*}
and \(\mathcal T_x\) denotes the set of directed spanning trees rooted at \(x\), i.e., each state \(y\neq x\) has a unique directed path to \(x\) along the edges of \(T\).

For each tree \(T\in\mathcal T_x\), \eqref{eq:ldp-step} gives
\[
-\eta\log
\left(
\prod_{(u\to v)\in T}P^{N,\eta}(u,v)
\right)
\longrightarrow
\sum_{(u\to v)\in T}\kappa^N(u,v),
\qquad \eta\to0.
\]
Since \(X^N\) is finite, the set \(\mathcal T_x\) is finite. Hence, the sum defining \(\tau_x^{N,\eta}\) is exponentially dominated by the minimum-cost tree. Define the stochastic potential
\[
V^N(x)
:=
\min_{T\in\mathcal T_x}
\sum_{(u\to v)\in T}\kappa^N(u,v).
\]
Then
\[
-\eta\log\tau_x^{N,\eta}\longrightarrow V^N(x),
\qquad \eta\to0.
\]
Therefore,
\[
\eta\log\frac{\mu^{N,\eta}(B)}{\mu^{N,\eta}(A)}
=
\eta\log\frac{\tau_B^{N,\eta}}{\tau_A^{N,\eta}}
\longrightarrow
V^N(A)-V^N(B).
\]
It remains to identify \(V^N(A)\) and \(V^N(B)\) with the corresponding escape costs.

We first show that
\begin{align*}
V^N(A)=\gamma_A^N
:=
C^N\!\bigl(B,\Exit^N(B)\bigr).
\end{align*}
The proof for \(V^N(B)=\gamma_B^N\) is symmetric.

For the upper bound, let
\begin{align*}
\phi=(\phi_0,\phi_1,\ldots,\phi_K)
\end{align*}
be a minimum-cost path from \(B\) to \(\Exit^N(B)\), with \(\phi_0=B\) and \(\phi_K\in\Exit^N(B)\). By definition,
\begin{align*}
C^N(\phi)=\gamma_A^N.
\end{align*}
Since \(\phi_K\notin \mathcal B_0^N(B)\), Assumption~\ref{ass:2attractor} (II) implies that \(C^N(\phi_K,A)=0\). Thus, there exists a zero-cost path from \(\phi_K\) to \(A\). Concatenating this zero-cost path with \(\phi\) gives a path from \(B\) to \(A\) with total cost \(\gamma_A^N\).

We now construct an \(A\)-rooted spanning tree. Include the edges of the above path from \(B\) to \(A\). For each remaining state \(x\in \mathcal B_0^N(A)\), attach \(x\) to \(A\) along a zero-cost path, stopping when the path first reaches the already constructed tree. For each remaining state \(x\notin \mathcal B_0^N(A)\), Assumption~\ref{ass:2attractor} (II) gives \(C^N(x,B)=0\); attach \(x\) along such a zero-cost path to \(B\), again stopping when the path first reaches the already constructed tree. Since all added edges have zero cost, this produces an \(A\)-rooted spanning tree with total cost \(\gamma_A^N\). Hence,
\[
V^N(A)\le \gamma_A^N.
\]

For the lower bound, let \(T\in\mathcal T_A\) be any \(A\)-rooted spanning tree. The unique directed path in \(T\) from \(B\) to \(A\) must leave the stability set \(\mathcal B_0^N(B)\). Let \(z\) be the first state on this path outside \(\mathcal B_0^N(B)\). Then \(z\in\Exit^N(B)\). Therefore, the cost of the portion of this path from \(B\) to \(z\) is at least
\begin{align*}
C^N\!\bigl(B,\Exit^N(B)\bigr)=\gamma_A^N.
\end{align*}
Since all transition costs are nonnegative, the total cost of \(T\) is at least \(\gamma_A^N\). Thus, we have
\[
V^N(A)\ge \gamma_A^N.
\]
Combining the upper and lower bounds gives \(V^N(A)=\gamma_A^N\). By the same argument,
\[
V^N(B)=\gamma_B^N
:=
C^N\!\bigl(A,\Exit^N(A)\bigr).
\]
Substituting these identities into the potential-ratio formula yields
\[
\lim_{\eta\to 0}
\eta\log\frac{\mu^{N,\eta}(B)}{\mu^{N,\eta}(A)}
=
\gamma_A^N-\gamma_B^N,
\]
which proves \eqref{eq:mass-ratio}.

It remains to establish the conclusions concerning stochastic
stability under Definition~\ref{def:sse}. By the assumed
vanishing-noise behavior of the noisy best-response functions,
every entry of \(P^{N,\eta}\) has a limit as \(\eta\to0\). Define
the limiting transition matrix by
\[
P^{N,0}(x,y)
:=
\lim_{\eta\to0}P^{N,\eta}(x,y).
\]
Since \(X^N\) is finite, \(P^{N,\eta}\to P^{N,0}\) entrywise, and
\(P^{N,0}\) is a stochastic matrix. Moreover, its positive-probability
off-diagonal transitions are precisely the zero-cost transitions of
the zero-noise dynamics.

Under Assumption~\ref{ass:2attractor}, the singleton sets
\(\{A\}\) and \(\{B\}\) are closed recurrent classes of \(P^{N,0}\).
They are also the only recurrent classes. Indeed, suppose that
\(\mathcal R\) were another recurrent class and let
\(x\in\mathcal R\). By
Assumption~\ref{ass:2attractor}\textup{(I)}, there is a zero-cost
path from \(x\) to either \(A\) or \(B\). This path consists of
positive-probability transitions of \(P^{N,0}\), contradicting the
closedness of \(\mathcal R\). Thus, the only recurrent classes of
\(P^{N,0}\) are \(\{A\}\) and \(\{B\}\).

We claim that
\(
\mu^{N,\eta}\bigl(X^N\setminus\{A,B\}\bigr)
\longrightarrow0
\), as \(\eta\to0\).
To prove this claim, consider an arbitrary sequence
\(\eta_m\to0\). Since the probability simplex on \(X^N\) is compact,
there is a subsequence, which we continue to denote by
\(\{\eta_m\}\), such that
\(
\mu^{N,\eta_m}\longrightarrow\bar\mu
\)
for some probability distribution \(\bar\mu\) on \(X^N\). Passing
to the limit in the stationarity equation
\(
\mu^{N,\eta_m}P^{N,\eta_m}
=
\mu^{N,\eta_m}
\)
gives
\(
\bar\mu P^{N,0}=\bar\mu.
\)
Thus, \(\bar\mu\) is a stationary distribution of \(P^{N,0}\).
Every stationary distribution of a finite Markov chain is supported
on its recurrent classes. Therefore,
\[
\bar\mu\bigl(X^N\setminus\{A,B\}\bigr)=0.
\]
Since this conclusion holds for every subsequential limit,
\(
\mu^{N,\eta}\bigl(X^N\setminus\{A,B\}\bigr)
\longrightarrow0.
\)
Equivalently,
\[
\mu^{N,\eta}(A)+\mu^{N,\eta}(B)
\longrightarrow1.
\]

Suppose now that
\(
\gamma_A^N<\gamma_B^N.
\)
By equation~\eqref{eq:mass-ratio},
\[
\eta\log
\frac{\mu^{N,\eta}(B)}{\mu^{N,\eta}(A)}
\longrightarrow
\gamma_A^N-\gamma_B^N<0.
\]
Hence,
\[
\frac{\mu^{N,\eta}(B)}{\mu^{N,\eta}(A)}
\longrightarrow0.
\]
Together with
\(
\mu^{N,\eta}(A)+\mu^{N,\eta}(B)
\longrightarrow1,
\)
this implies
\[
\mu^{N,\eta}(A)\longrightarrow1
\qquad\text{and}\qquad
\mu^{N,\eta}(B)\longrightarrow0.
\]
Therefore, \(A\) is uniquely stochastically stable according to
Definition~\ref{def:sse}.

The case
\(\gamma_B^N<\gamma_A^N\)
is symmetric and gives
\[
\mu^{N,\eta}(B)\longrightarrow1.
\]
Thus, \(B\) is uniquely stochastically stable.
\end{proof}

\section{Auxiliary Lemmas}

\begin{lemma}\label{lem:Slogit}

The function \(\bar S_{\mathrm{logit}}\) defined in \eqref{eq:bar-s-logit} is continuous on \(\mathbb R\) and strictly monotone unless \(|a-b|=|c-d|\), in which case it is constant.
Moreover,
\[
\lim_{\beta\to-\infty}\bar S_{\mathrm{logit}}(\beta)
=
-\min\{a,b\}+\min\{c,d\},
\]
and
\[
\lim_{\beta\to+\infty}\bar S_{\mathrm{logit}}(\beta)
=
-\max\{a,b\}+\max\{c,d\}.
\]
\end{lemma}

\begin{proof}
Continuity of \(\bar S_{\mathrm{logit}}\) follows from the continuity of
\((\beta,x)\mapsto \Delta_\beta(x)\) on \(\mathbb R\times[0,1]\).

For \(p\le q\), define
\[
\psi_{p,q}(\beta):=
\begin{cases}
\dfrac{q e^{\beta q}-p e^{\beta p}}{e^{\beta q}-e^{\beta p}}
-\dfrac1\beta,
& \beta\neq 0,\ p<q,\\[2ex]
\dfrac{p+q}{2}, & \beta=0,\ p<q,\\[1.2ex]
p, & p=q.
\end{cases}
\]
The apparent singularity at \(\beta=0\) is removable, since for \(p<q\),
\[
\frac{q e^{\beta q}-p e^{\beta p}}{e^{\beta q}-e^{\beta p}}
-\frac1\beta
\longrightarrow
\frac{p+q}{2}
\qquad\text{as }\beta\to0.
\]
Thus, \(\psi_{p,q}\) is continuous on \(\mathbb R\).

For \(\beta\neq0\), integrating the logarithmic expression for
\(\Delta_\beta\) gives
\[
\bar S_{\mathrm{logit}}(\beta)
=
-\psi_{\min\{a,b\},\max\{a,b\}}(\beta)
+
\psi_{\min\{c,d\},\max\{c,d\}}(\beta).
\]
By continuity, the same identity holds at \(\beta=0\). In particular,
\[
\bar S_{\mathrm{logit}}(0)
=
-\frac{a+b}{2}
+
\frac{c+d}{2}
=
-\frac{a+b-c-d}{2}.
\]

We now prove monotonicity. If \(p<q\), then
\[
\psi'_{p,q}(\beta)
=
\frac1{\beta^2}
-
\frac{(q-p)^2 e^{\beta(q-p)}}{(e^{\beta(q-p)}-1)^2},
\qquad \beta\neq0.
\]
If \(p=q\), then \(\psi_{p,p}\equiv p\), so \(\psi'_{p,p}=0\).

Hence, for \(\beta\neq0\),
\[
\bar S_{\mathrm{logit}}'(\beta)
=
G_\beta(|a-b|)-G_\beta(|c-d|),
\]
where
\[
G_\beta(z):=
\begin{cases}
\dfrac{z^2 e^{\beta z}}{(e^{\beta z}-1)^2}, & z>0,\\[2ex]
\dfrac1{\beta^2}, & z=0.
\end{cases}
\]
The value at \(z=0\) is the continuous extension, since
\[
\lim_{z\downarrow0}
\frac{z^2 e^{\beta z}}{(e^{\beta z}-1)^2}
=
\frac1{\beta^2}.
\]
Thus, the degenerate cases \(|a-b|=0\) or \(|c-d|=0\) are handled.

For each fixed \(\beta\neq0\), the function
\(z\mapsto G_\beta(z)\) is strictly decreasing on \([0,\infty)\). This is the same calculation used in the proof of Lemma~\ref{lem:xstar-beta-monotonicity}. Therefore, the sign of
\(\bar S_{\mathrm{logit}}'(\beta)\) is constant on
\((-\infty,0)\cup(0,\infty)\).
Since \(\bar S_{\mathrm{logit}}\) is continuous at \(\beta=0\), it follows that
\(\bar S_{\mathrm{logit}}\) is monotone on all of \(\mathbb R\). It is strictly
monotone unless \(|a-b|=|c-d|\), in which case it is constant. 

Finally, if \(p<q\), then
\[
\lim_{\beta\to-\infty}\psi_{p,q}(\beta)=p,
\qquad
\lim_{\beta\to+\infty}\psi_{p,q}(\beta)=q,
\]
while \(\psi_{p,p}\equiv p\). Substituting these limits into the expression for
\(\bar S_{\mathrm{logit}}(\beta)\) gives
\[
\lim_{\beta\to-\infty}\bar S_{\mathrm{logit}}(\beta)
=
-\min\{a,b\}+\min\{c,d\},
\]
and
\[
\lim_{\beta\to+\infty}\bar S_{\mathrm{logit}}(\beta)
=
-\max\{a,b\}+\max\{c,d\}.
\]
\end{proof}

\begin{lemma}\label{lem:Rpos}
Fix \(p\in\{1,2\}\). Suppose that, for population \(p\),
\[
a^p\ge d^p
\quad\text{and}\quad
b^p\ge c^p,
\]
with at least one inequality strict. Then, for every \(\beta^p\in\mathbb R\),
if \(\Delta(\cdot)=\Delta_{\beta^p}^p(\cdot)\) denotes the ERM payoff difference
in \eqref{eq:erm-delta}, we have
\[
\Delta(1-z)>-\Delta(z),
\qquad \forall z\in[0,1].
\]
Equivalently, \(\Delta(z)+\Delta(1-z)>0\) for all \(z\in[0,1]\).
\end{lemma}

\begin{proof}
Fix \(\beta^p\in\mathbb R\), and define
\[
R_{\beta^p}(z):=\Delta(z)+\Delta(1-z),
\qquad z\in[0,1].
\]
We first consider the case \(\beta^p\neq 0\). Let
\[
A=e^{\beta^p a^p},\quad
B=e^{\beta^p b^p},\quad
C=e^{\beta^p c^p},\quad
D=e^{\beta^p d^p}.
\]
Then
\begin{align*}
\beta^p R_{\beta^p}(z)
=
\log\!\left(
\frac{(zA+(1-z)B)\bigl((1-z)A+zB\bigr)}
     {(zC+(1-z)D)\bigl((1-z)C+zD\bigr)}
\right).
\end{align*}
Writing \(t=z(1-z)\in[0,\tfrac14]\), we obtain
\[
\exp\!\bigl(\beta^p R_{\beta^p}(z)\bigr)
=
\frac{AB+t(A-B)^2}{CD+t(C-D)^2}.
\]

Since \(a^p\ge d^p\) and \(b^p\ge c^p\), with at least one strict inequality,
we have
\[
a^p+b^p>c^p+d^p.
\]
If \(\beta^p>0\), then \(AB>CD\), and also
\[
A+B>C+D,
\]
because \(A\ge D\), \(B\ge C\), and at least one inequality is strict. Hence, the
affine function
\[
t\mapsto
\bigl[AB+t(A-B)^2\bigr]
-
\bigl[CD+t(C-D)^2\bigr]
\]
is positive at \(t=0\). At \(t=\tfrac14\), it equals
\[
\frac{(A+B)^2-(C+D)^2}{4},
\]
which is also positive. Therefore, it is positive for all
\(t\in[0,\tfrac14]\). Thus
\[
\exp\!\bigl(\beta^p R_{\beta^p}(z)\bigr)>1,
\]
and since \(\beta^p>0\), we get \(R_{\beta^p}(z)>0\).

If \(\beta^p<0\), the inequalities reverse: \(AB<CD\) and \(A+B<C+D\).
The same affine-function argument gives
\[
\exp\!\bigl(\beta^p R_{\beta^p}(z)\bigr)<1.
\]
Since now \(\beta^p<0\), it again follows that \(R_{\beta^p}(z)>0\).

Finally, when \(\beta^p=0\), the ERM payoffs reduce to expected payoffs, so
\[
\Delta(z)
=
z(a^p-c^p)+(1-z)(b^p-d^p).
\]
Hence
\[
R_0(z)
=
\Delta(z)+\Delta(1-z)
=
(a^p+b^p)-(c^p+d^p)>0.
\]
Therefore, \(R_{\beta^p}(z)>0\) for every \(\beta^p\in\mathbb R\) and every
\(z\in[0,1]\), which is equivalent to
\[
\Delta(1-z)>-\Delta(z),
\]
completing the proof.
\end{proof}




\end{appendices}


\bibliography{bibliography}

\end{document}